\documentclass[a4paper,UKenglish]{lipics-v2018}
\usepackage{amssymb}
\usepackage{amsthm}
\usepackage{cite}
\usepackage{thmtools}
\usepackage{xparse}
\usepackage{mathtools}
\usepackage[basic]{complexity}
\usepackage{microtype}

\theoremstyle{plain}

\theoremstyle{definition}

\DeclareMathOperator{\mypolylog}{polylog}

\newcommand{\N}{\mathbb{N}}

\newclass{\lcl}{LCL}
\newclass{\local}{LOCAL}

\newcommand{\A}{\ensuremath{\mathcal A}}
\newcommand{\G}{\ensuremath{\mathcal G}}
\newcommand{\F}{\ensuremath{\mathcal F}}

\newcommand{\dirnext}{\ensuremath{\mathsf {Next}}}
\newcommand{\dirprev}{\ensuremath{\mathsf {Prev}}}
\newcommand{\error}{\ensuremath{\mathsf {Error}}}
\newcommand{\lhead}{\ensuremath{\mathsf {Head}}}
\newcommand{\lstate}{\ensuremath{\mathsf {State}}}
\newcommand{\ltape}{\ensuremath{\mathsf {Tape}}}
\newcommand{\lexemptu}{\ensuremath{\mathsf {Exempt_U}}}
\newcommand{\ldimension}{\ensuremath{\mathsf {Dimension}}}

\newcommand{\lunbalanced}{\ensuremath{\mathsf {Unbalanced}}}
\newcommand{\lfinished}{\ensuremath{\mathsf {Exempt_M}}}
\newcommand{\lin}{\ensuremath{\mathsf {in}}}
\newcommand{\lout}{\ensuremath{\mathsf {out}}}

\newcommand{\modelname}{Linear Bounded Automaton}
\newcommand{\shortname}{\mbox{\sf{LBA}}}

\newcommand{\checkradius}{r}
\newcommand{\tnew}{\tau}
\newcommand{\told}{\tau_{\operatorname{\mathsf {orig}}}}
\newcommand{\tc}{c}
\newcommand{\ball}[1]{\operatorname{\mathsf{Ball}}_{#1}}
\newcommand{\dg}[1]{d_{#1}}
\newcommand{\dgg}[2]{d_{#1}^{#2}}
\DeclareMathOperator{\del}{\mathsf{Del}}

\newcommand{\deltree}{\mathcal{T}}
\newcommand{\paths}{\mathcal{P}}
\newcommand{\spaths}{\mathcal{Q}}
\newcommand{\lpump}{\ell_{\operatorname{\mathsf {pump}}}}
\DeclareMathOperator{\pump}{\mathsf {Pump}}
\newcommand{\aorig}{\mathcal{A}}
\newcommand{\pumpmult}{B}
\DeclareMathOperator{\replace}{\mathsf {Replace}}
\newcommand{\eqrg}{\stackrel{*}{\sim}}
\newcommand{\labeling}{\mathcal{L}}
\newcommand{\dist}[2]{\operatorname{dist}(#1,#2)}
\DeclareMathOperator{\virt}{\mathsf {Virt}}
\DeclareMathOperator{\typepath}{Type}
\DeclareMathOperator{\typetree}{Class}

\newcommand{\lsetgrid}{\ensuremath{\mathcal{L}^\mathsf {grid}}}
\newcommand{\lbalabels}{\ensuremath{\mathcal{L}^\mathsf {encoding}}}
\newcommand{\lsetunbalanced}{\ensuremath{\mathcal{L}^\mathsf {unbalanced}}}
\newcommand{\lseterror}{\ensuremath{\mathcal{L}^\mathsf {error}}}

\newcommand{\figscale}{0.75}

\title{Almost Global Problems in the \boldmath \texorpdfstring{\local{}}{LOCAL} Model}

\author{Alkida Balliu}{University of Freiburg, Germany\footnote{current affiliation}\\ Aalto University, Finland\footnote{where most of this work was done}} {alkida.balliu@cs.uni-freiburg.de} {} {}

\author{Sebastian Brandt}{ETH Zurich, Switzerland$^{1,2}$} {brandts@ethz.ch}{}{}

\author{Dennis Olivetti}{University of Freiburg, Germany$^1$\\ Aalto University, Finland$^2$} {dennis.olivetti@cs.uni-freiburg.de} {} {}

\author{Jukka Suomela}{Aalto University, Finland$^{1,2}$} {jukka.suomela@aalto.fi} {} {}

\authorrunning{A. Balliu, S. Brandt, D. Olivetti, J. Suomela}

\Copyright{Alkida Balliu, Sebastian Brandt, Dennis Olivetti, Jukka Suomela}
\subjclass{
	Theory of computation $\rightarrow$ Distributed computing models,
	Theory of computation $\rightarrow$ Complexity classes
}


\funding{This work was supported in part by the Academy of Finland, Grant 285721.}

\nolinenumbers
\hideLIPIcs
\ArticleNo{1}

\hypersetup{
    pdftitle={Almost Global Problems in the LOCAL Model},
    pdfauthor={Alkida Balliu, Sebastian Brandt, Dennis Olivetti, Jukka Suomela}
}

\begin{document}
	
\maketitle

\begin{abstract}
	The landscape of the distributed time complexity is nowadays well-understood for subpolynomial complexities. When we look at deterministic algorithms in the \local{} model and locally checkable problems (\lcl{}s) in bounded-degree graphs, the following picture emerges:
	\begin{itemize}
		\item There are lots of problems with time complexities of $\Theta(\log^* n)$ or $\Theta(\log n)$.
		\item It is not possible to have a problem with complexity between $\omega(\log^* n)$ and $o(\log n)$.
		\item In \emph{general graphs}, we can construct \lcl{} problems with infinitely many complexities between $\omega(\log n)$ and $n^{o(1)}$.
		\item In \emph{trees}, problems with such complexities do not exist.
	\end{itemize}
	However, the high end of the complexity spectrum was left open by prior work. In general graphs there are \lcl{} problems with complexities of the form $\Theta(n^\alpha)$ for any rational $0 < \alpha \le 1/2$, while for trees only complexities of the form $\Theta(n^{1/k})$ are known. No \lcl{} problem with complexity between $\omega(\sqrt{n})$ and $o(n)$ is known, and neither are there results that would show that such problems do not exist. We show that:
	\begin{itemize}
		\item In \emph{general graphs}, we can construct \lcl{} problems with infinitely many complexities between $\omega(\sqrt{n})$ and $o(n)$.
		\item In \emph{trees}, problems with such complexities do not exist.
	\end{itemize}
	Put otherwise, we show that any \lcl{} with a complexity $o(n)$ can be solved in time $O(\sqrt{n})$ in trees, while the same is not true in general graphs.
\end{abstract}

\section{Introduction}

Recently, in the study of \emph{distributed graph algorithms}, there has been a lot of interest on \emph{structural complexity theory}: instead of studying the distributed time complexity of specific graph problems, researchers have started to put more focus on the study of \emph{complexity classes} in this context.

\subparagraph{LCL problems.}

A particularly fruitful research direction has been the study of distributed time complexity classes of so-called $\lcl$ problems (locally checkable labellings). We will define $\lcl$s formally in Section~\ref{subsec:def-lcl}, but the informal idea is that $\lcl$s are graph problems in which \emph{feasible solutions can be verified by checking all constant-radius neighbourhoods}. Examples of such problems include vertex colouring with $k$ colours, edge colouring with $k$ colours, maximal independent sets, maximal matchings, and sinkless orientations.

$\lcl$s play a role similar to the class $\NP$ in the centralised complexity theory: these are problems that would be easy to solve with a \emph{nondeterministic} distributed algorithm---guess a solution and verify it in $O(1)$ rounds---but it is not at all obvious what is the distributed time complexity of solving a given $\lcl$ problem with \emph{deterministic} distributed algorithms.

\subparagraph{Distributed structural complexity.}

In the classical (centralised, sequential) complexity theory one of the cornerstones is the \emph{time hierarchy theorem}~\cite{HS95}. In essence, it is known that giving more time always makes it possible to solve more problems.  Distributed structural complexity is fundamentally different: there are various \emph{gap results} that establish that there are no $\lcl$ problems with complexities in a certain range. For example, it is known that there is no $\lcl$ problem whose deterministic time complexity on bounded-degree graphs is between $\omega(\log^* n)$ and $o(\log n)$ \cite{chang16exponential}.

Such gap results have also direct applications: we can \emph{speed up} algorithms for which the current upper bound falls in one of the gaps. For example, it is known that $\Delta$-colouring in bounded-degree graphs can be solved in $\mypolylog n$ time \cite{panconesi95delta}. Hence $4$-colouring in 2-dimensional grids can be also solved in $\mypolylog n$ time. But we also know that in 2-dimensional grids there is a gap in distributed time complexities between $\omega(\log^* n)$ and $o(\sqrt{n})$ \cite{Brandt2017}, and therefore we know we can solve $4$-colouring in $O(\log^* n)$ time.

The ultimate goal here is to identify all such gaps in the landscape of distributed time complexity, for each graph class of interest.

\subparagraph{State of the art.}

Some of the most interesting open problems at the moment are related to \emph{polynomial complexities in trees}. The key results from prior work are:
\begin{itemize}
	\item In bounded-degree trees, for each positive integer $k$ there is an $\lcl$ problem with time complexity $\Theta(n^{1/k})$ \cite{Chang2019}.
	\item In bounded-degree graphs, for each rational number $0 < \alpha \le 1/2$ there is an $\lcl$ problem with time complexity $\Theta(n^\alpha)$ \cite{Balliu2018stoc}.
\end{itemize}
However, there is no separation between trees and general graphs in the polynomial region. Furthermore, we do not have any $\lcl$ problems with time complexities $\Theta(n^\alpha)$ for any $1/2 < \alpha < 1$.

\subparagraph{Our contributions.}

This work resolves both of the above questions. We show that:
\begin{itemize}
	\item In bounded-degree graphs, for each rational number $1/2 < \alpha < 1$ there is an $\lcl$ problem with time complexity $\Theta(n^\alpha)$.
	\item In bounded-degree trees, there is no $\lcl$ problem with time complexity between $\omega(\sqrt{n})$ and $o(n)$.
\end{itemize}
Hence whenever we have a slightly sublinear algorithm, we can always speed it up to $O(\sqrt{n})$ in trees, but this is not always possible in general graphs.

\subparagraph{Key techniques.}

We use ideas from the classical centralised complexity theory---e.g.\ Turing machines and regular languages---to prove results in distributed complexity theory.

The key idea for showing that there are \lcl{}s with complexities $\Theta(n^\alpha)$ in bounded-degree graphs is that we can take any \emph{linear bounded automaton} $M$ (a Turing machine with a bounded tape), and construct an $\lcl$ problem $\Pi_M$ such that the \emph{distributed} time complexity of $\Pi$ is a function of the \emph{sequential} running time of $M$. Prior work \cite{Balliu2018stoc} used a class of \emph{counter machines} for a somewhat similar purpose, but the construction in the present work is much simpler, and Turing machines are more convenient to program than the counter machines used in the prior work.

To prove the gap result, we heavily rely on Chang and Pettie's \cite{Chang2019} ideas: they show that one can relate $\lcl$ problems in trees to regular languages and this way generate equivalent subtrees by ``pumping''. However, there is one fundamental difference:
\begin{itemize}
	\item Chang and Pettie first construct certain \emph{universal} collections of tree fragments (that do not depend on the input graph), use the existence of a fast algorithm to show that these fragments can be labelled in a convenient way, and finally use such a labelling to solve any given input efficiently.
	\item We work directly with the \emph{specific} input graph, expand it by ``pumping'', and apply a fast algorithm there directly.
\end{itemize}
Many speedup results make use of the following idea: given a graph with $n$ nodes, we pick a much \emph{smaller} value $n' \ll n$ and lie to the algorithm that we have a tiny graph with only $n'$ nodes \cite{chang16exponential,Brandt2017}. Our approach essentially reverses this: given a graph with $n$ nodes and an algorithm $\A$, we pick a much \emph{larger} value $n' \gg n$ and lie to the algorithm that we have a huge graph with $n'$ nodes.

\subparagraph{Open problems.}

Our work establishes a gap between $\Theta(n^{1/2})$ and $\Theta(n)$ in trees. The next natural step would be to generalise the result and establish a gap between $\Theta(n^{1/(k+1)})$ and $\Theta(n^{1/k})$ for all positive integers $k$.

\section{Model and related work}
As we study LCL problems, a family of problems defined on \emph{bounded-degree graphs}, we assume that our input graphs are of degree at most $\Delta$, where $\Delta = O(1)$ is a known constant. Each input graph $G=(V,E)$ is simple, connected, and undirected; here $V$ is the set of nodes and $E$ is the set of edges, and we denote by $n=|V|$ the total number of nodes in the input graph.

\subsection{Model of computation}\label{subsec:def.model}
The model considered in this paper is the well studied \local{} model~\cite{Peleg2000,Linial1992}. In the \local{} model, each node $v \in V$ of the input graph $G$ runs the same deterministic algorithm.
The nodes are labelled with unique $O(\log n)$-bit identifiers, and initially each node knows only its own identifier, its own degree, and the total number of nodes $n$.

Computation proceeds in synchronous rounds. At each round, each node
\begin{itemize}
	\item sends a message to its neighbours (it may be a different message for different neighbours),
	\item receives messages from its neighbours,
	\item performs some computation based on the received messages. 
\end{itemize}
In the \local{} model, there is no restriction on the size of the messages or on the computational power of a node. The time complexity of an algorithm is measured as the number of communication rounds that are required such that every node is able to stop and produce its local output. Hence, after $t$ rounds in the \local{} model, each node can gather the information in the network up to distance $t$ from it. In other words, in $t$ rounds a node can gather all information within its $t$-radius neighbourhood, where the $t$-radius neighbourhood of a node $v$ is the subgraph containing all nodes at distance at most $t$ from $v$ and all edges incident to nodes at distance at most $t-1$ from $v$ (including the inputs given to these nodes). Also, in $t$ rounds, the information outside the $t$-radius neighbourhood of a node $v$ cannot reach~$v$. This means that a $t$-round algorithm running in the \local{} model can be seen as a function that maps all possible $t$-radius neighbourhoods to the outputs. Notice that, in the \local{} model, every problem can be solved in \emph{diameter} number of rounds, where the diameter of a graph $G$ is defined as the largest hop-distance among any pair of nodes in $G$. In fact, in diameter time each node can gather all information there is in the whole graph and solve the problem locally.

\subsection{Locally checkable labellings}\label{subsec:def-lcl}
Locally checkable labelling problems (\lcl{}s) were introduced in the seminal work of Naor and Stockmeyer \cite{Naor1995}. Informally, \lcl{}s are graph problems defined on bounded-degree graphs (i.e., graphs where the maximum degree is constant with respect to the number of nodes), where nodes have as input a label from a constant-size set of input labels, and they must produce an output from a constant-size set of output labels. The validity of these output labels is determined by a set of local constraints.

\subparagraph{Formal definition.}

Let $\F$ be the family of bounded-degree graphs. An \lcl{} is defined as a tuple $\Pi = (\Sigma_{\lin}, \Sigma_{\lout},\allowbreak C, r)$ as follows.
\begin{itemize}
	\item $\Sigma_{\lin}$ and $\Sigma_{\lout}$ are constant-size sets of labels;
	\item $r$ is an arbitrary constant (called the checkability radius of the problem);
	\item $C$ is a set of graphs where 
	\begin{itemize}
		\item each graph $H\in C$ is centred at some node $v$,
		\item the distance of $v$ from all other nodes in $H$, i.e., the radius of $v$, is at most $r$,
		\item each node $u$ is labelled with a pair $(i(u), o(u))\in \Sigma_{\lin}\times \Sigma_{\lout}$.
	\end{itemize} 
\end{itemize}

\subparagraph{An example.}

An example of an \lcl{} problem is vertex $3$-colouring, where $\Sigma_{\lin}=\{\bot\}$, $\Sigma_{\lout}=\{1,2,3\}$, $r=1$, and $C$ is defined as all graphs of radius $1$ in $\F$ such that each node has a colour in $\Sigma_{\lout}$ that is different from the ones of its neighbours.

\subparagraph{Solving a problem.}

In general, solving an \lcl{} means the following. We are given a graph $G=(V, E)\in\F$ and an input assignment $i\colon V\to\Sigma_{\lin}$. The goal is to produce an output assignment $o\colon V\to\Sigma_{\lout}$. Let $B(v)$ be the subgraph of $G$ induced by nodes of distance at most $r$ from $v$, augmented with the inputs and outputs assigned by $i$ and $o$. The output assignment is valid if and only if, for each node $v\in V$, we have $B(v)\in C$. In that case, we call $(G,i,o)$ a valid configuration. 

This can be adapted to a distributed setting in a straightforward manner: if we are solving an \lcl{} in the \local{} model with a distributed algorithm $\A$, the input graph $G=(V, E)\in\F$ is the communication network, each node $v$ initially knows only its own part of the input $i(v)\in\Sigma_{\lin}$, and when algorithm $\A$ stops, each node $v$ has to know its own part of output $o(v)\in\Sigma_{\lout}$. The local outputs have to form a valid configuration $(G,i,o)$.

\subparagraph{Distributed time complexity.}

The distributed time complexity of an \lcl{} problem $\Pi$ in a graph family $\F$ is the pointwise smallest $t\colon \N \to \N$ such that there is a distributed algorithm $\A$ that solves $\Pi$ in $t(n)$ communication rounds in any graph $G \in \F$ with $n$ nodes, for any $n \in \N$, and for any input labelling of $G$.

\subparagraph{Distributed verifiers.}

Above, we have defined an \lcl{} as a set $C$ of correctly labelled subgraphs. Equivalently, we could define an \lcl{} in terms of a \emph{verifier} $\A'$. A verifier is a distributed algorithm that receives both $i$ and $o$ as inputs, runs for $r$ communication rounds, and then each node $v$ outputs either `accept' or `reject'. We require that the output of $\A'$ does not depend on the ID assignment or on the size of the input graph, but only on the structure of $G$ and input and output labels in the $r$-radius neighbourhood of $v$. Now we say that $(G,i,o)$ is a valid configuration if all nodes output `accept'.

This is equivalent to the above definition, as in $r$ communication rounds each node $v$ can gather all information within distance $r$, and nothing else. Hence $\A'$ can output `accept' if $B(v) \in C$; equivalently, the output of any such algorithm $\A'$ defines a set $C$ of correctly labelled neighbourhoods.

If $\A$ solves an \lcl{} problem $\Pi$ in time $t(n)$, and $\A'$ is the verifier for $\Pi$, then by definition the composition of $\A$ and $\A'$ is a distributed algorithm that runs in $t(n) + r$ rounds and always outputs `accept' everywhere. It is important to note that, while the output of algorithm $\A$ may depend on the ID assignment that nodes have, the output of verifier $\A'$ must not depend on the ID assignment or on the size of the graph.

\subsection{Related work}

\subparagraph{Cycles and paths.} \lcl{} problems are fully understood in the case of cycles and paths. In these graphs it is known that there are \lcl{} problems having complexities $O(1)$, e.g.\ trivial problems, $\Theta(\log^* n)$, e.g.\ $3$-vertex colouring, and $\Theta(n)$, e.g.\ $2$-vertex colouring~\cite{Linial1992,cole86deterministic}. Chang, Kopelowitz, and Pettie~\cite{chang16exponential} showed two automatic speedup results: any $o(\log^* n)$-time algorithm can be converted into an $O(1)$-time algorithm; any $o(n)$-time algorithm can be converted into an $O(\log^* n)$-time algorithm. They also showed that randomness does not help in comparison with deterministic algorithms in cycles and paths.

\subparagraph{Oriented grids.} Brandt et al.~\cite{Brandt2017} studied \lcl{} problems on oriented grids, showing that, as in the case of cycles and paths, the only possible complexities of \lcl{}s are $O(1)$, $\Theta(\log^* n)$, and $\Theta(n)$, on $n\times n$ grids, and it is also known that randomness does not help \cite{chang16exponential, GHK18}. However, while it is decidable whether a given \lcl{} on cycles can be solved in $t$ rounds in the \local{} model~\cite{Naor1995,Brandt2017}, it is not the case for oriented grids~\cite{Brandt2017}.

\subparagraph{Trees.} Although well studied, \lcl{}s on trees are not fully understood yet. Chang and Pettie \cite{Chang2019} show that any $n^{o(1)}$-time algorithm can be converted into an $O(\log n)$-time algorithm. In the same paper they show how to obtain \lcl{} problems on trees having deterministic and randomized complexity of $\Theta(n^{1/k})$, for any integer $k$. However, it is not known if there are problems of complexities between $o(n^{1/k})$ and $\omega(n^{1/(k+1)})$. Regarding decidability on trees, given an \lcl{} it is possible to decide whether it has complexity $O(\log n)$ or $n^{O(1)}$ \cite{Chang2019}. In other words, it is possible to decide on which side of the gap between $\omega(\log n)$ and $n^{o(1)}$ an \lcl{} lies, but it is still an open question if we can decide whether a given \lcl{} has complexity $O(\log^* n)$ or $\Omega(\log n)$.

\subparagraph{General graphs.} Another key direction of research is understanding \lcl{}s on general (bounded-degree) graphs. Using the techniques presented by Naor and Stockmeyer \cite{Naor1995}, it is possible to show that any $o(\log \log^*n)$-time algorithm can be sped up to $O(1)$ rounds. It is known that there are \lcl{} problems with complexities $\Theta(\log^* n)$~\cite{barenboim16sublinear,barenboim14distributed,panconesi01simple,fraigniaud16local} and $\Theta(\log n)$~\cite{Brandt2016,chang16exponential,ghaffari17distributed}. On the other hand, Chang et al.~\cite{chang16exponential} showed that there are no $\lcl{}$ problems with deterministic complexities between $\omega(\log^* n)$ and $o(\log n)$. It is known that there are problems (for example, $\Delta$-colouring) that require $\Omega(\log n)$ rounds~\cite{Brandt2016,chang18complexity}, for which there are algorithms solving them in $O(\mypolylog n)$ rounds~\cite{panconesi95delta}. Until very recently, it was thought that there would be many other gaps in the landscape of complexities of \lcl{} problems in general graphs. Unfortunately, it has been shown in \cite{Balliu2018stoc} that this is not the case: it is possible to obtain \lcl{}s with numerous different deterministic time complexities, including $\Theta( \log^{\alpha} n )$ and $\Theta( \log^{\alpha} \log^* n )$  for any $\alpha \ge 1$, $2^{\Theta( \log^{\alpha} n )}$, $\smash{2^{\Theta( \log^{\alpha} \log^* n )}}$, and $\Theta((\log^* n)^{\alpha})$ for any $\alpha \le 1$, and $\Theta(n^{\alpha})$ for any $\alpha < 1/2$ (where $\alpha$ is a positive rational number).

\section{Near-linear complexities in general graphs}

In this section we show that there are \lcl{}s with time complexities in the spectrum between $\omega(\sqrt{n})$ and $o(n)$. To show this result, we prove that we can take any linear bounded automaton (\shortname) $M$, that is, a Turing machine with a tape of a bounded size, and an integer $i \geq 2$, and construct an \lcl{} problem $\Pi_M^i$, such that the distributed time complexity of $\Pi_M^i$ is related to the choice of $i$ and to the sequential running time of $M$ when starting from an empty tape.

In particular, given an \shortname{} $M$, we will define a family of graphs, that we call \emph{valid instances}, where nodes are required to output an encoding of the execution of $M$. An \lcl{} must be defined on any (bounded-degree) graph, without any promise on the graph structure, thus, we will define the \lcl{} $\Pi_M^i$ by requiring nodes to either prove that the given instance is not a valid instance, or to output a correct encoding of the execution of $M$ if the instance is a valid one. 
The manner in which the execution has to be encoded ensures that the complexity of the \lcl{} $\Pi_M^i$ will depend on the running time of the \shortname{} $M$, and by constructing \shortname{}s with suitable running times, we can show our result. The key idea here is that we will use valid instances to prove a lower bound on the time complexity of our \lcl{}s, and we will prove that adding all the other instances does not make the problem harder.

\subparagraph{A simplified example.}

For example, consider an \shortname{} $M$ that encodes a unary counter, starting with the all-$0$ bit string and terminating when the all-$1$ bit string is reached. Clearly, the running time of $M$ is linear in the length of the tape. We can represent the full execution of $M$ using $2$ dimensions, one for the tape and the other for time, and we can encode this execution on a $2$-dimensional grid. See Figure~\ref{fig:example-unaryCounter} for an illustration. Notice that the length of the \emph{time} dimension of this grid depends only on the length of the \emph{tape} dimension and on $M$, and for a unary counter the length of the time dimension will always be the same as the length of the tape dimension. The \lcl{} $\Pi_M^2$ will be defined such that valid instances are $2$-dimensional grids with balanced dimensions $\sqrt{n} \times \sqrt{n}$ ($n$ nodes in total), and the idea is that, given such a grid, nodes are required to output an encoding of the full execution of $M$, and this would require $\Theta(n^{1/2})$ rounds (since, in order to determine their output bit, certain nodes will have to determine the bit string they are part of, which in turn requires seeing the far end of the grid where the all-$0$ bit string is encoded). 

In order to obtain \lcl{}s with time complexity $\Theta(n^\alpha)$, where $1/2<\alpha< 1$, we define valid instances in a slightly different manner. We consider grids with $i > 2$ dimensions, and we let nodes choose where to encode the execution of $M$. We will allow nodes to choose an arbitrary dimension to use as the tape dimension, but, for technical reasons, we restrict nodes to use dimension $1$ as the time dimension.
The idea here is that, if the size of dimensions $2,\dotsc,i$ is not the same, nodes can minimize their running time by coordinating and picking the dimension $j$ (different from $1$) of smallest length as the tape dimension, and encode the execution of $M$ using dimension $1$ for time and dimension $j$ for the tape. Thereby we ensure that in a worst-case instance all dimensions except dimension $1$ have the same length. Also, if a grid has strictly fewer or more than $i$ dimensions, it will not be part of the family of valid instances. In other words, our \lcl{}s can be solved faster if the input graph does not look like a multidimensional grid with the correct number $i$ of dimensions. Then, by using \shortname{}s with different running times, and by choosing different values for $i$, we will prove our claim.

\begin{figure}
	\centering
	\includegraphics[width=0.30\textwidth]{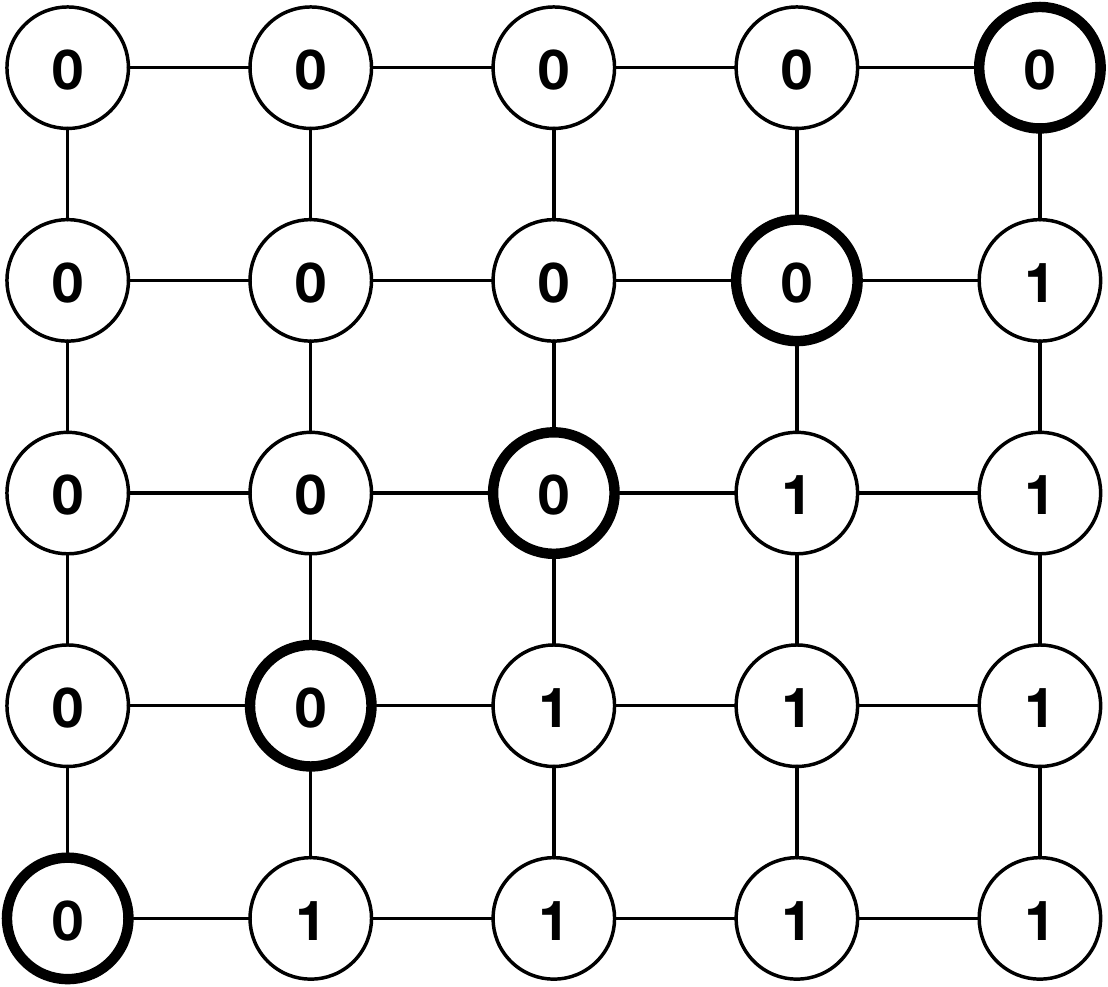}
	\caption{An example of a correct encoding of an \shortname{} execution in a grid; the horizontal dimension represents time, while the vertical dimension represents the position on the tape. The \shortname{} executes a unary counter, and the marked nodes represent the position of the head.}
	\label{fig:example-unaryCounter}
\end{figure}

\subparagraph{Technicalities.}
The process of defining an \lcl{} that requires a solution to encode an \shortname{} execution as output needs a lot of care.
 Denote with $M$ an \shortname{}. At a high level, our \lcl{} problems are designed such that the encoding of the execution of $M$ needs to be done only on valid instances. In other words, the \lcl{} $\Pi_M^i$ will satisfy the following:
\begin{itemize}
	\item If the graph looks like a multidimensional grid with $i$ dimensions, then the output of some nodes that are part of some $2$-dimensional surface of the grid must properly encode the execution of $M$.
	\item Otherwise, nodes are exempted from encoding the execution of $M$, but they must produce as output a proof of the fact that the instance is invalid.
\end{itemize}
This second point is somehow delicate, as nodes may try to ``cheat'' by claiming that a valid instance is invalid. Also, recall that in an $\lcl{}$ it has to be possible to check the validity of a solution \emph{locally}, that is, there must exist a constant time distributed algorithm such that, given a correct solution, it outputs \emph{accept} on all nodes, while given an invalid solution, it outputs \emph{reject} on at least one node. To deal with these issues, we define our \lcl{}s as follows:
\begin{itemize}
	\item Valid instances are multidimensional grids with inputs. This input is a locally checkable proof that the given instance is a valid one, that is, nodes can check in constant time if the input is valid, and if it is not valid, at least one node must detect an error (for more on locally checkable proofs we refer the reader to \citen{Goos2016}). On these valid instances, nodes can only output a correct encoding of $M$.
	\item On invalid instances, nodes must be able to prove that the input does not correctly encode a proof of the validity of the instance. This new proof must also be locally checkable.
\end{itemize}
Thus, we will define two kinds of locally checkable proofs (each using only a constant number of labels, since we will later need to encode them in the definition of the \lcl{}s): the first is given as input to nodes, and it should satisfy that only on valid instances all nodes see a correct proof, while on invalid instances at least one node sees some error; the second is given as output from nodes, and it should satisfy that all nodes see a correct proof only if there exists a node that, in constant time, notices some error in the first proof.

Hence, we will define \lcl{}s that are solvable on any graph by either encoding the execution of the \shortname{}, or by proving that the current graph is not a valid instance, where this last part is done by showing that the proof given as input is invalid on at least one node.

\subsection{Roadmap}
We will now describe the high level structure of this section. We will start by formally introducing Linear Bounded Automata in Section~\ref{ssec:lba}. 

We will then introduce multidimensional grids in Section~\ref{ssec:grids}: these will be the hardest instances for our \lcl{}s. These grids will be labelled with some input, and we will provide a set of local constraints for this input such that, if these constraints are satisfied for all nodes, then the graph must be a multidimensional grid of some desired dimension $i$ (or belong to a small class of grid-like graphs that we have to deal with separately). Also, for any multidimensional grid of dimension $i$, it should be possible to assign these inputs in a valid manner. In other words, we design a locally checkable proof mechanism for the family of graphs of our interest, and every node will be able to verify if constraints hold by just inspecting their $3$-radius neighbourhood, and essentially constraints will be valid on all nodes if and only if the graph is a valid multidimensional grid (Sections \ref{subsec:gridlab} and \ref{app:local-checkability}).

Next we will define what are valid outputs on multidimensional grids with the desired number $i$ of dimensions. The idea is that nodes must encode the execution of some \shortname{} $M$ on the surface spanned by $2$ dimensions. Nodes will be able to choose which dimension to use as the \emph{tape} dimension, but they will be forced to use dimension $1$ as the \emph{time} dimension. The reason why we do not allow nodes to choose both dimensions is that, in order to obtain complexities in the $\omega(\sqrt{n})$ spectrum, we will need the time dimension to be  $\omega(\sqrt{n})$, but in any grid with at least $3$ dimensions, the smallest two dimensions are always $O(\sqrt{n})$. For example, consider an \shortname{} $M$ that encodes a unary $5$-counter, that is, a list of $5$ unary counters, such that when one counter overflows, the next one is incremented. The running time of $M$ is $\Theta(B^5)$, where $B$ is the length of the tape. The worst case instance for the problem $\Pi^3_M$ will be a $3$-dimensional grid, where dimensions $2$ and $3$ will have equal size $n^{1/7}$ and dimension $1$ will have size $n^{5/7}$. In such an instance, nodes will be required to encode the execution of $M$ using either dimension $2$ or $3$ as tape dimension, and $1$ as time dimension---note that the size of dimension $1$ as a function of the size of dimension $2$ (or $3$) matches the running time of $M$ as a function of $B$. Thus, the complexity of $\Pi^3_M$ will be $\Theta(n^{5/7})$, as nodes will need to see up to distance $\Theta(n^{5/7})$ in dimension $2$ (or $3$). If we do not force nodes to choose dimension $1$ as time, then nodes can always find two dimensions of size $O(n^{1/2})$, and we would not be able to obtain problems with complexity $\omega(n^{1/2})$.

We will start by handling grids that are \emph{unbalanced} in a certain way, that is, where dimension $1$ is too small compared to all the others (Section~\ref{subsec:unbalanced}). In this case, deviating from the above, we allow nodes to produce some different output that can be obtained without spending much time (this is needed to ensure that our \lcl{}s do not get too hard on very unbalanced grids). Then, we define what the outputs must be on valid grids that are not unbalanced (Section~\ref{subsec:machineencoding}). Each node must produce an output such that the ensemble of the outputs of nodes encodes the execution of a certain \shortname{}. In particular, we define a set of output labels and a set of constraints, such that the constraints are valid on all nodes if and only if the output of the nodes correctly encodes the execution of the \shortname{}.

We define our \lcl{}s in Section~\ref{subsec:construction}. We provide a set of output labels, and constraints for these labels, that nodes can use to prove that the given graph is not a valid multidimensional grid, where the idea is that nodes can produce pointers that form a directed path that must end on a node for which the grid constraints are not satisfied. Our \lcl{} will then be defined such that nodes can either:
\begin{itemize}
	\item produce an encoding of the execution of the given \shortname, or
	\item prove that dimension $1$ is too short, or
	\item prove that there is an error in the grid structure.
\end{itemize}
All this must be done carefully, ensuring that nodes cannot claim that there is an error in valid instances, and always allowing nodes to produce a proof of an error if the instance is invalid. Also, we cannot require all nodes to consistently choose one of the three options, since that may require too much time. So we must define constraints such that, for example, it is allowed for some nodes to produce a valid encoding of the execution of the \shortname, and at the same time it is allowed for some other nodes to prove that there is an error in the input proof (that maybe the first group of nodes did not see).

Finally, in Section~\ref{ssec:timecomplexity} we will show upper and lower bounds for our \lcl{}s, and in Section~\ref{subsec:instantiate} we show how these results imply the existence of a family of \lcl{}s that have complexities in the range between $\omega(\sqrt{n})$ and $o(n)$.

\subparagraph{Remarks.}

To avoid confusion, we point out that we will (implicitly) present two very different distributed algorithms in this section:
\begin{itemize}
	\item First, we define a specific \lcl{} problem $\Pi_M^i$. Recall that any \lcl{} problem can be interpreted as a constant-time distributed algorithm $\A'$ that verifies that $(G,i,o)$ is a valid configuration. We do not give $\A'$ explicitly here, but we will present a list of constraints that can be checked in constant time by each node. This is done in Section~\ref{subsec:lcldef}.
	\item Second, we prove that the distributed complexity of $\Pi_M^i$ is $\Theta(t(n))$, for some $t$ between $\omega(\sqrt{n})$ and $o(n)$. To show this, we will need a pair of matching upper and lower bounds, and to prove the upper bound, we explicitly construct a distributed algorithm $\A$ that solves $\Pi_M^i$ in $O(t(n))$ rounds, i.e., $\A$ takes $(G,i)$ as input and produces some $o$ as output such that $(G,i,o)$ is a valid configuration that makes $\A'$ happy. This is done in Section~\ref{sssec:upperbound}.
\end{itemize}
Note that the specific details of $\Pi_M^i$ as such are not particularly interesting; the interesting part is that $\Pi_M^i$ is an \lcl{} problem (in the strict formal sense) and its distributed time complexity is between $\omega(\sqrt{n})$ and $o(n)$. As we will see in Section~\ref{sec:trees} such problems do not exist in trees.

\subsection{Linear bounded automata}\label{ssec:lba}
A \modelname{} (\shortname) $M$ consists of a Turing machine that is executed on a tape with bounded size, able to recognize the boundaries of the tape~\cite[p.~225]{HU79}. We consider a simplified version of \shortname{}s, where the machine is initialized with an empty tape (no input is present). We describe this simplified version of \shortname{}s as a $5$-tuple $M = (Q,q_0,f,\Gamma,\delta)$, where:
\begin{itemize}
	\item $Q$ is a finite set of states;
	\item $q_0 \in Q$ is the initial state;
	\item $f \in Q$ is the final state;
	\item $\Gamma$ is a finite set of tape alphabet symbols, containing a special symbol $b$ (\emph{blank}), and two special symbols, $L$ and $R$, called \emph{left} and \emph{right} markers;
	\item $\delta \colon Q\setminus\{f\} \times \Gamma \to Q \times \Gamma \times \{-,\leftarrow,\rightarrow\}$ is the transition function.
\end{itemize}
The tape is initialized in the following way:
\begin{itemize}
	\item the first cell contains the symbol $L$;
	\item the last cell contains the symbol $R$;
	\item all the other cells contain the symbol $b$.
\end{itemize}
The head is initially positioned on the cell containing the symbol $L$.  Then, depending on the current state and the symbol present on the current position of the tape head, the machine enters a new state, writes a symbol on the current position, and moves to some direction.

In particular, we describe the transition function $\delta$ by a finite set of $5$-tuples $(s_0,t_0,s_1,t_1,d)$, where:
\begin{enumerate}
	\item The first $2$ elements specify the input:
	\begin{itemize}
		\item $s_0$ indicates the current state;
		\item $t_0$ indicates the tape content on the current head position.
	\end{itemize} 
	\item The remaining $3$ elements specify the output:
	\begin{itemize}
		\item $s_1$ is the new state;
		\item $t_1$ is the new tape content on the current head position;
		\item $d$ specifies the new position of the head:
		\begin{itemize}
			\item `$\rightarrow$' means that the head moves to the next cell in direction towards $R$;
			\item `$\leftarrow$' indicates that the head moves to the next cell in direction towards $L$;
			\item `$-$' means the head does not move.
		\end{itemize}
	\end{itemize}
\end{enumerate}
The transition function must satisfy that it cannot move the head beyond the boundaries $L$ and $R$, and the special symbols $L$ and $R$ cannot be overwritten. If $\delta$ is not defined on the current state and tape content, the machine terminates.

By fixing a machine $M$ and by changing the size of the tape $B$ on which it is executed, we obtain different running times for the machine, as a function of $B$. We denote by $T_M(B)$ the \emph{running time} of an \shortname{} $M$ on a tape of size $B$. For example, it is easy to design a machine $M$ that implements a binary counter, that starts from a tape containing all $0$s and terminates when the tape contains all $1$s, and this machine has a running time $T_M(B) = \Theta(2^B)$.

Also, it is possible to define a \emph{unary $k$-counter}, that is, a list of $k$ unary counters (where each one counts from $0$ to $B-1$ and then overflows and starts counting from $0$ again) in which when a counter overflows, the next is incremented. It is possible to obtain running times of the form $T_M(B) = \Theta(B^k)$ by carefully implementing these counters (for example by using a single tape of length $B$ to encode all the $k$ counters at the cost of using more machine states and tape symbols).

The reason why we consider \shortname{}s is that they fit nicely with the \lcl{} framework, that requires local checkability using a constant number of labels. The definition of an \shortname{} $M$ does not depend on the tape size, that is, the description of $M$ is constant compared to $B$. Also, by encoding the execution of $M$ using two dimensions, one for the tape and the other for time, we obtain a structure that is locally checkable: the correctness of each constant size neighbourhood of this two dimensional surface implies global correctness of the encoding.

\subsection{Grid structure}\label{ssec:grids}
In order to obtain \lcl{}s for the general graphs setting, we need our \lcl{}s to be defined on any (bounded-degree) graph, and not only on some family of graphs of our interest. That is, we cannot assume any promise on the family of graphs where the \lcl{} requires to be solved. The challenge here is that we can easily encode \shortname{}s only on grids, but not on general graphs. 

Thus, we will define our \lcl{}s in a way that there is a family of graphs, called \emph{valid instances}, where nodes are actually required to output the encoding of the execution of a specific \shortname{}, while on other instances nodes are exempted from doing so, but in this case they are required to prove that the graph is not a valid instance. The intuition here is that valid instances will be \emph{hard} instances for our \lcl{}s, meaning that when we will prove a lower bound for the time complexity of our \lcl{}s we will use graphs that are part of the family of valid instances. Then, when we will prove upper bounds, we will show that our \lcl{}s are always solvable, even in graphs that are invalid instances, and the time required to solve the problem in these instances is not more than the time required to solve the problem in the lower bound graphs that we provided.

We will now make a first step in defining the family of valid instances, by formally defining what a grid graph is.

Let $i \ge 2$ and $d_1,\ldots,d_i$ be positive integers.
The set of nodes of an \emph{$i$-dimensional grid graph} $\G$ consists of all $i$-tuples $u=(u_1,\ldots,u_i)$ with $0 \le u_j \le d_j$ for all $1 \le j \le i$. 
We call $u_1,\ldots,u_i$ the \emph{coordinates} of node $u$ and $d_1,\ldots,d_i$ the \emph{sizes} of the \emph{dimensions} $1,\ldots,i$. Let $u$ and $v$ be two arbitrary nodes of $\G$. There is an edge between $u$ and $v$ if and only if $||u - v||_1=1$, i.e., all coordinates of $u$ and $v$ are equal, except one that differs by $1$. Figure~\ref{fig:grid2} depicts a grid graph with $3$ dimensions.

Notice that nodes do not know their position in the grid, or, for incident edges, which dimension they belong to. In fact, nodes do not even know if the graph where they currently are is actually a grid!
At the beginning nodes only know the size $n$ of the graph, and everything else must be discovered by exploring the neighbourhood.

\begin{figure}
	\centering
	\includegraphics[width=0.45\textwidth]{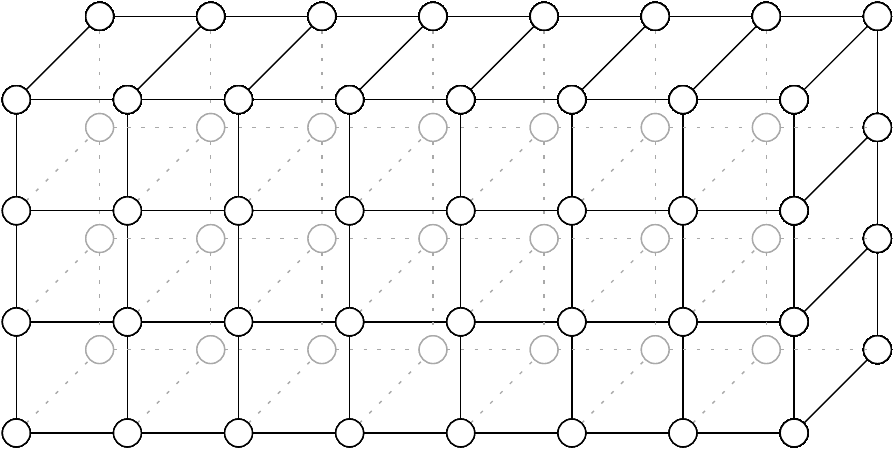}
	\caption{An example of a 3-dimensional grid graph.}
	\label{fig:grid2}
\end{figure}

\subsubsection{Grid labels}\label{subsec:gridlab}
As previously discussed, \lcl{}s must be well defined on any (bounded-degree) graph, and we want to define our \lcl{}s such that, if a graph is not a valid instance, then it must be easy for nodes to prove such a claim, where easy here means that the time required to produce such a proof is not larger than the time required to encode the execution of the machine in the worst possible valid instance of the same size. For this purpose, we need to \emph{help} nodes to produce such a proof. The idea is that a valid instance not only must be a valid grid graph, but it must also include a proof of being a valid grid graph. Thus, we will define a constant size set of labels, that will be given as input to nodes, and a set of constraints, such that, if a graph is a valid grid graph, then it is possible to assign a labelling that satisfies the constraints, while if the graph is not a valid grid graph, then it should be not possible to assign a labelling such that constraints are satisfied on all nodes (at least one node must note some inconsistency). Informally, in the \lcl{}s that we will define, such a labelling will be provided as input to the nodes, and nodes will be able to prove that a graph is invalid by just pointing to a node where the input labelling does not satisfy these constraints.

For the sake of readability, instead of defining a set of labels and a set of local constraints that characterize valid grid graphs, we start by defining what is a valid label assignment to grid graphs, in a non-local manner. Then, we will show how to turn this to a set of locally checkable constraints. Unfortunately, we will not be able to prove that if labels satisfy these local constraints on all nodes, then the graph is actually a grid. Instead, the set of graphs that satisfy these constraints for all nodes will include a small family of additional graphs, that are graphs that locally look like a grid everywhere, but globally they are not valid grid graphs. For example, toroidal grids will be in this family. As we will show, the weaker statement that we prove will be enough for our purposes.

We now present a valid labelling for valid grid graphs. Each edge $e=\{u, v\}$ is assigned two labels $L_u(e)$ and $L_v(e)$, one for each endpoint. Label $L_u(e)$ is chosen as follows:
\begin{itemize}
	\item $L_u(e) = (\dirnext,j)$ if $v_j - u_j = 1$;
	\item $L_u(e)= (\dirprev,j)$ if $u_j - v_j = 1$.
\end{itemize}
Label $L_v(e)$ is chosen analogously. We define $\lsetgrid$ to be the set of all possible input labels, i.e.,
\[
	\lsetgrid =
		\bigl\{ (\dirprev,j) \bigm| 1 \leq j \leq i \bigr\} \cup
		\bigl\{ (\dirnext,j) \bigm| 1 \leq j \leq i \bigr\}.
\]

If we want to focus on a specific label of some edge $e$ and it is clear from the context which of the two edge labels is considered, we may refer to it simply as the label of $e$.

We call the unique node that does not have any incident edge labelled $(\dirprev,j)$ \emph{origin}.
Equivalently, we could define the origin directly as the node $(0, 0, \dots, 0)$, but we want to emphasize that each node of a grid graph can infer whether it is the origin, simply by checking its incident labels.

In Section~\ref{subsec:lcldef}, the defined grid labels will appear as edge input labels in the definition of the new \lcl{} problems we design.
In the formal definition of an \lcl{} problem (see Section~\ref{subsec:def-lcl}), input labels are assigned to \emph{nodes}; however this is not a problem: that we label \emph{edges} in our grid graphs is just a matter of convenience; we could equally well assign the labels to nodes instead of edges (and, for that matter, combine all labels of a node into a single label).
The same holds for the \emph{output} labels that are part of the definitions of the \lcl{} problems in Section~\ref{subsec:lcldef}.
Furthermore, we could also equally well encode the labels in the graph structure.
Hence, all new time complexities presented in Section~\ref{subsec:instantiate} can also be achieved by \lcl{} problems without input labels (a family of problems frequently considered in the \local{} model literature).
From now on, with \emph{grid graph} we denote a grid graph \emph{with a valid labelling}.

\subsubsection{Local checkability}\label{app:local-checkability}
As previously discussed, we want to make sure that if the graph is not a valid grid graph, then at least one node can detect this issue in constant time. Hence, we are interested in a \emph{local characterisation} of grid graphs.
Given such a characterisation, each node can check locally whether the input graph has a valid grid structure in its neighbourhood.
As it turns out, such a characterization is not possible, since there are non-grid graphs that look like grid graphs locally everywhere, but we can come sufficiently close for our purposes.
In fact, we will specify a set of local constraints that characterise a class of graphs that contains all grid graphs of dimension $i$ (and a few other graphs).
All the constraints depend on the $3$-radius neighbourhood of the nodes, so for each input graph not contained in the characterised graph class, at least one node can detect in $3$ rounds (in the \local{} model) that the graph is not a grid graph.

For any node $v$ and any sequence $L_1, L_2, \dotsc, L_p$ of edge labels, let $z_v(L_1,L_2,\dotsc,L_p)$ denote the node reached by starting in $v$ and following the edges with labels $L_1,L_2,\dotsc,L_p$.
If at any point during traversing these edges there are $0$ or at least $2$ edges with the currently processed label, $z_v(L_1,L_2,\dotsc,L_p)$ is not defined (this may happen, since nodes need to be able to check if the constraints hold also on graphs that are invalid grid graphs).
Let $i \ge 2$. The full constraints are given below:
\begin{enumerate}
	\item Basic properties of the labelling. For each node $v$ the following hold:\label{item:link-first} 
	\begin{itemize}
		\item Each edge $e$ incident to $v$ has exactly one ($v$-sided) label $L_v(e)$, and for some $1 \le j \le i$ we have
		\[L_v(e) \in \bigl\{ (\dirprev,j),\, (\dirnext,j) \bigr\}.\]
		\item For any two edges $e, e'$ incident to $v$, we have \[L_v(e) \neq L_v(e').\]
		\item For any $1 \le j \le i$, there is at least one edge $e$ incident to $v$ with
		\[L_v(e) \in \bigl\{ (\dirprev,j),\, (\dirnext,j) \bigr\}.\]
	\end{itemize}
	\item Validity of the grid structure. For each node $v$ the following hold:\label{item:link-last}
	\begin{itemize}
		\item For any incident edge $e=\{ v, u \}$, we have that
		\begin{align*}
		L_u(e) = (\dirprev,j) &\text{ if } L_v(e) = (\dirnext,j), \\
		L_u(e) = (\dirnext,j) &\text{ if } L_v(e) = (\dirprev,j).
		\end{align*}
		\item Let $1 \le j, k \le i$ such that $j \neq k$. Also, let $e=\{ v, u \}$ and $e'=\{ v, u' \}$ be edges with the following $v$-sided labels:
		\begin{align*}
		L_v(e) &\in \bigl\{(\dirprev,j),\, (\dirnext,j)\bigr\}, \\
		L_v(e') &\in \bigl\{(\dirprev,k),\, (\dirnext,k)\bigr\}.
		\end{align*}
		Then node $u$ has an incident edge $e''$ with label $L_u(e'') = L_v(e')$, and $u'$ has an incident edge $e'''$ with label $L_{u'}(e''') = L_v(e)$. Moreover, the two other endpoints of $e''$ and $e'''$ are the same node, i.e., $z_u(L_u(e'')) = z_{u'}(L_{u'}(e'''))$.
	\end{itemize}
\end{enumerate}
It is clear that $i$-dimensional grid graphs satisfy the given constraints, but as observed above, the converse statement is not true. (As a side note for the curious reader, we mention that the converse statement can be transformed into a correct (and slightly weaker) statement by adding the small (non-local) condition that the considered graph contains a node not having any incident edge labelled with some $(\dirprev,j)$, for all dimensions $j$. However, due to its non-local nature, we cannot make use of such a condition.)

\subsubsection{Unbalanced grid graphs}\label{subsec:unbalanced}
In Section~\ref{app:local-checkability}, we saw the basic idea behind ensuring that non-grid graphs are not among the hardest instances for the \lcl{} problems we construct.
In this section, we will study the ingredient of our \lcl{} construction that guarantees that grid graphs where the dimensions have ``wrong'' sizes are not worst-case instances.
More precisely, we want that the hardest instances for our \lcl{} problems are grid graphs with the property that there is at least one dimension $2 \le j \le i$ whose size is not larger than the size of dimension $1$.
In the following, we will show how to make sure that \emph{unbalanced} grid graphs, i.e., grid graphs that do not have this property, allow nodes to find a valid output without having to see too far.
In a sense, in the \lcl{}s that we construct, one possible valid output is to produce a proof that the grid is unbalanced in a wrong way, and since the validity of an output assignment for an \lcl{} must be locally checkable, we want such a proof to be locally checkable as well. 

Thus, in the \lcl{}s that we will define, nodes of an arbitrary graph will be provided with some input labelling that encodes a (possibly wrong) proof that claims that the current graph is a valid grid graph. Then, if the graph does not look like a grid, nodes can produce a locally checkable proof that claims that this input proof is wrong. Instead, if the graph does look like a grid, but this grid appears to be unbalanced in some undesired way, nodes can produce a locally checkable proof about this fact.

More formally, consider a grid graph with $i$ dimensions of sizes $d_1,\ldots,d_i$. If $d_1 < d_j$ for all $2 \le j \le i$, the following output labelling is regarded as correct in any constructed \lcl{} problem:
\begin{itemize}
	\item For all $0 \le t \le d_1$, node $v = (v_1,\ldots, v_i)$ satisfying $v_1 = \dotsc = v_i = t$ is labelled $\lunbalanced{}$.
	\item All other nodes are labelled $\lexemptu{}$.
\end{itemize}
This labelling is clearly locally checkable, i.e., it can be described as a collection of local constraints:
Each node $v$ labelled $\lunbalanced{}$ checks that
\begin{enumerate}
	\item its two ``diagonal neighbours''
	\begin{align*}
	u &= z_v((\dirprev,1), (\dirprev,2), \dots, (\dirprev,i)) \mbox{ and} \\
	w &= z_v((\dirnext,1), (\dirnext,2), \dots, (\dirnext,i)),
	\end{align*} 
	both exist (i.e., are defined) and are both labelled $\lunbalanced{}$, or
	\item $w$ exists and is labelled $\lunbalanced{}$ and $v$ has no incident edge labelled $(\dirprev,j)$, or
	\item $u$ exists and is labelled $\lunbalanced{}$ and $v$ has an incident edge labelled $(\dirnext,j)$ for all $2 \le j \le i$, but no incident edge labelled $(\dirnext,1)$\label{constraint:unbalanced}.
\end{enumerate} 
Condition \ref{constraint:unbalanced} ensures that the described diagonal chain of $\lunbalanced{}$ labels terminates at a node whose first coordinate is $d_1$ (i.e., the maximal possible value for the coordinate corresponding to dimension $1$), but whose second, third, \ldots, coordinate is strictly smaller than $d_2, d_3, \dots$, respectively, thereby guaranteeing that grid graphs that are not unbalanced do not allow the output labelling specified above. Finally, the origin checks that it is labelled $\lunbalanced{}$, in order to prevent the possibility that each node simply outputs $\lexemptu{}$. We refer to Figure~\ref{fig:unbalanced} for an example of an unbalanced $2$-dimensional grid and its labelling. We define $\lsetunbalanced$ to be the set containing $\{ \lunbalanced, \lexemptu \}$.

\begin{figure}
	\centering
	\includegraphics[width=0.25\textwidth]{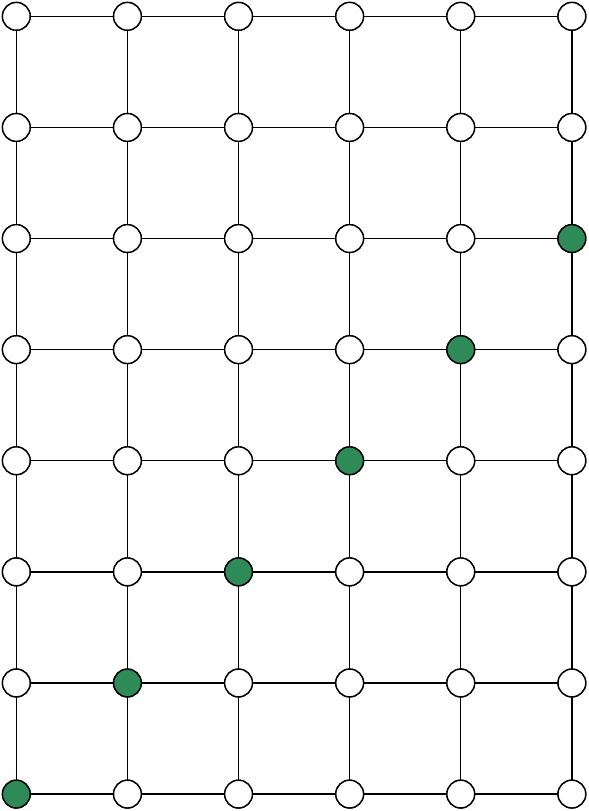}
	\caption{An example of an unbalanced grid with $2$ dimensions; the nodes in green are labelled with \lunbalanced{}, while the white nodes are labelled with \lexemptu.}
	\label{fig:unbalanced}
\end{figure}

\subsection{Machine encoding}\label{subsec:machineencoding}
After examining the cases of the input graph being a non-grid graph or an unbalanced grid graph, in this section, we turn our attention towards the last remaining case: that is the input graph is actually a grid graph for which there is a dimension with size smaller than or equal to the size of dimension $1$.
In this case, we require the nodes to work together to create a global output that is determined by some \shortname{} $M$.
Essentially, the execution of $M$ has to be written (as node outputs) on a specific part of the grid graph.
In order to formalise this relation between the desired output and the \shortname{} $M$, we introduce the notion of an \emph{$M$-encoding graph} in the following section.

\subsubsection{Labels}
Let $M$ be an \shortname{}, and consider the execution of $M$ on a tape of size $B$. Let $S_\ell = (s_\ell,h_\ell,t_\ell)$ be the whole state of $M$ after step $\ell$, where $s_\ell$ is the machine internal state, $h_\ell$ is the position of the head, and $t_\ell$ is the whole tape content. The content of the cell in position $y \in \{ 0, \dotsc, B-1 \}$ after step $\ell$ is denoted by $t_\ell[y]$. We denote by $(x,y)_k$ the node $v=(v_1,\ldots, v_i)$ having $v_1 = x$, $v_k = y$, and $v_j = 0$ for all $j \not\in \{1,k\}$. An (output-labelled) grid graph of dimension $i$ is an \emph{$M$-encoding graph} if there exist a tape size $B$ and a dimension $2 \le k \le i$ satisfying the following conditions. 
\begin{description}
	\item[(C1)] $d_k$ is equal to $B-1$.
	\item[(C2)] For all $0 \le x \le \min\{ T_M(B), d_1 \}$ and all $0 \le y \le B-1$, it holds that:
	\begin{description}
		\item[(a)] Node $(x,y)_k$ is labelled with $\ltape(t_x[y])$.
		\item[(b)] Node $(x,y)_k$ is labelled with $\lstate(s_x)$.
		\item[(c)] Node $(x,h_x)_k$ is labelled with $\lhead$.
		\item[(d)] Node $(x,y)_k$ is labelled with $\ldimension(k)$.
	\end{description}
	\item[(C3)] All other nodes are labelled with $\lfinished$.
\end{description}
Intuitively, the $2$-dimensional surface expanding in dimensions $1$ and $k$ (having all the other coordinates equal to $0$), encodes the execution of $M$.
More precisely, the nodes $(x,0)_k, (x,1)_k, \dots, (x,B-1)_k$ together represent the state of the tape at time $x$, i.e., dimension $1$ constitutes the time axis whereas the tape itself is unrolled along dimension $k$.
In particular, the nodes $(0,1)_k, (0,2)_k, \dots, (0, B-2)_k$ representing the (inner part of the) tape at the beginning of the computation are all labelled with the blank symbol $b$ (or, if we want to be very precise, $\ltape(b)$), the nodes $(0,0)_k, (0,1)_k, \dots$ representing the left end of the tape (at different points in time) are labelled with the left marker $L$, and the nodes $(B-1,0)_k, (B-1,1)_k, \dots$ representing the right end of the tape are labelled with the right marker $R$.  We define $\lbalabels$ to be the set of all possible output labels used to define an $M$-encoding graph.

\subsubsection{Local checkability}\label{subsec:locallba}
In order to force nodes to output labels that turn the input grid graph into an $M$-encoding graph, we must be able to describe Conditions (C1)--(C3) in the form required by an \lcl{} problem, i.e., as locally checkable constraints.
In the following, we provide such a description, showing that the nodes can check whether the graph correctly encodes the execution of a given \shortname{}~$M$. 

\subparagraph{Constraint (LC1):} Each node $v$ is labelled with either $\lfinished$ or $\ldimension(k)$ for exactly one $2 \le k \le i$. In the former case, node $v$ has no other labels, in the latter case, $v$ additionally has some $\ltape$ and some $\lstate$ label, and potentially the label $\lhead$, but no other labels.

\subparagraph{Constraint (LC2):} The origin has label $\ldimension(k)$, for some $2\le k \le i$.

\subparagraph{Constraint (LC3):} If a node $v$ labelled $\ldimension(k)$ for some $2\le k \le i$ has an incident edge $e$ labelled with $L_v(e) = (\dirprev,j)$, then $j=1$ or $j=k$. Moreover, for each node $v$ labelled $\ldimension(k)$, nodes $z_v((\dirprev,1))$, $z_v((\dirprev,k))$ and $z_v((\dirnext,k))$ (provided they are defined) are also labelled $\ldimension(k)$.

\subparagraph{Constraint (LC4):} For each node $v$ labelled $\ldimension(k)$ for some $2\le k \le i$, the following hold:
	\begin{enumerate}
		\item If $v$ does not have an incident edge labelled $(\dirprev,1)$, then
		\begin{description}
			\item[(a)] if $v$ does not have an incident edge labelled $(\dirprev,k)$, then it must have labels $\lhead$ and $\ltape(L)$;
			\item[(b)] if $v$ does not have an incident edge labelled $(\dirnext,k)$, then it must have label $\ltape(R)$;
			\item[(c)] if $v$ has an incident edge labelled $(\dirprev,k)$ and an incident edge labelled $(\dirnext,k)$, then it has label $\ltape(b)$;
			\item[(d)] $v$ has label $\lstate(q_0)$;
			\item[(e)] if $q_0 \neq f$, then node $z_v((\dirnext,1))$ (if defined) is labelled $\ldimension(k)$.
		\end{description}
		\item If $v$ has an incident edge labelled $(\dirprev,1)$, $v$ has labels $\lstate(q)$ and $\ltape(t)$, and $z_v((\dirprev,1))$ has labels $\lstate(q')$ and $\ltape(t')$, then
		\begin{description}
			\item[(f)] $q' \neq f$;
			\item[(g)] if $z_v((\dirprev,1))$ is labelled with $\lhead$, then $q$ and $t$ are derived from $q'$ and $t'$ according to the specifications of the \shortname{} $M$, and the new position of the head is either on $v$ itself, or on $z_v((\dirprev,k))$, or on $z_v((\dirnext,k))$, depending on $M$;
			\item[(h)] otherwise, $t=t'$ and the nodes $z_v((\dirprev,k))$ and $z_v((\dirnext,k))$ (if defined) are labelled $\lstate(q)$;
			\item[(i)] if $q \neq f$, then node $z_v((\dirnext,1))$ (if defined) is labelled $\ldimension(k)$.
		\end{description}
	\end{enumerate}

\subparagraph{Correctness.} It is clear that an $M$-encoding graph satisfies Constraints (LC1)--(LC4). Conversely, we want to show that any graph satisfying Constraints  (LC1)--(LC4) is an $M$-encoding graph. We start by setting $B := d_k+1$, which already implies Condition (C1).

\emph{Constraints (LC1)--(LC3)} ensure that there is a $2$-dimensional surface $\mathcal{S}$ on which the execution of $M$ is encoded: (LC3) ensures that for any node labelled $\ldimension(k)$, all coordinates except potentially those corresponding to dimension $1$ and $k$ are $0$, i.e., each node labelled $\ldimension(k)$ is of the form $(x,y)_k$ for some $x,y$. Moreover, according to (LC3), whenever some $(x,y)_k$ is labelled $\ldimension(k)$, then also all $(x', y')_k$ that satisfy $x' \leq x$ are labelled $\ldimension(k)$ and, in particular, the origin is also labelled $\ldimension(k)$. Since, by (LC1), no node (in particular, the origin) is labelled $\ldimension(k)$ for more than one $k$, it follows that there is at most one $k$ for which there exist nodes with label $\ldimension(k)$, and these nodes are exactly the nodes $(x,y)_k$ for which $x$ does not exceed some threshold value (which as we will see will be exactly $\min\{ T_M(B), d_1 \}$). (LC2) ensures that this value is at least $1$; in particular, there are nodes that are not labelled $\lfinished$. (LC1) ensures that all nodes not labelled $\ldimension(k)$ are labelled $\lfinished$.

\emph{Constraints (LC4a)--(LC4d)} ensure that the \shortname{} $M$ is initialized correctly: (LC4a)--(LC4c) ensure that $(0,0)_k$ is labelled with $\ltape(L) = \ltape(t_0[0])$, $(0,y)_k$ is labelled with $\ltape(b) = \ltape(t_0[y])$ for each $1 \leq y \leq B-2$, and $(0,B-1)_k$ is labelled with $\ltape(R) = \ltape(t_0[B-1])$, which implies Condition (C2a) for $x=0$. Similarly, (LC4a) also ensures (C2c) for $x=0$, and (LC4d) ensures (C2b) for $x=0$.

\emph{Constraints (LC4e)--(LC4i)} ensure a correct execution of each step of $M$, and that nodes on $\mathcal{S}$ output $\lfinished$ only after the termination state of $M$ is reached: Constraints (LC4e) and (LC4i) ensure that the threshold value for $y$ up to which all $(x,y)_k$ are labelled with $\ldimension(k)$ is at least $T_M(B)$, unless $T_M(B) > d_1$, in which case the threshold value is $d_1$. (LC4f) ensures that the threshold value does not exceed $T_M(B)$, thereby implying Conditions (C2d) and (C3). Here we use the observation derived from (LC1) that all nodes not labelled $\ldimension(k)$ are labelled $\lfinished$. (LC4g) and (LC4h) imply that if the nodes $(x,0)_k, (x,1)_k, \dots, (x,B-1)_k$ encode the state of the computation after step $x$, and the corresponding machine internal state is not the final state, then also the nodes $(x+1,0)_k, (x+1,1)_k, \dots, (x+1,B-1)_k$ encode the state of the computation after step $x+1$. As we already observed above that (LC4a)--(LC4d) ensure that $(0,0)_k, (0,1)_k, \dots, (0,B-1)_k$ encode the initial state of the computation, we obtain by induction (and our obtained knowledge about the threshold value) that (C2a)--(C2c) hold for all $0 \le x \le \min\{ T_M(B), d_1 \}$.

\subsection{\boldmath\texorpdfstring{\lcl{}}{LCL} construction}\label{subsec:construction}
Fix an integer $i \ge 2$, and let $M$ be an \shortname{} with running time $T_M$. As we do not fix a specific size of the tape, $T_M$ can be seen as a function that maps the tape size $B$ to the running time $T_M(B)$ of the \shortname{} executed on a tape of size $B$. We now construct an \lcl{} problem $\Pi_M^i$ with complexity related to $T_M$. Note that $\Pi_M^i$ depends on the choice of $i$. The general idea of the construction is that nodes can either:
\begin{itemize}
	\item produce a valid encoding of the execution of $M$, or
	\item prove that dimension $1$ is too short, or
	\item prove that there is an error in the (grid) graph structure.
\end{itemize}
We need to ensure that on balanced grid graphs it is not easy to claim that there is an error, while allowing an efficient solution on invalid graphs, i.e., graphs that contain a local error (some invalid label), or a global error (a grid structure that wraps, or dimension $1$ too short compared to the others). 

\subsubsection{\boldmath\texorpdfstring{\lcl{}}{LCL} problem \texorpdfstring{$\Pi_M^i$}{Pi\^{}i\_M}}\label{subsec:lcldef}

Formally, we specify the \lcl{} problem $\Pi_M^i$ as follows. The input label set for $\Pi_M^i$ is the set \lsetgrid{} of labels used in the grid labelling (see Section~\ref{subsec:gridlab}). The possible output labels are the following:
\begin{enumerate}
	\item The labels from $\lbalabels$ (see Section~\ref{subsec:machineencoding}).
	\item The labels from $\lsetunbalanced$ (see Section~\ref{subsec:unbalanced}).
	\item The set of \emph{error labels} $\lseterror$. This set is defined to contain the \emph{error label} $\error$, and \emph{error pointers}, i.e., all possible pairs $(s, r)$, where $s$ is either $(\dirnext,j)$ or $(\dirprev,j)$ for some $1\le j \le i$, and $r \in \{ 0, 1 \}$ is a bit whose purpose it is to distinguish between two different types of error pointers, \emph{type 0} pointers and \emph{type 1} pointers.
\end{enumerate}
Intuitively, nodes that notice that there is/must be an error in the grid structure, but are not allowed to output $\error$ because the grid structure is valid \emph {in their local neighbourhood}, can point in the direction of an error.
However, the nodes have to make sure that the error pointers form a chain that actually ends in an error.
In order to make the proofs in this section more accessible, we distinguish between the two types of error pointers mentioned above; roughly speaking, type 0 pointers will be used by nodes that (during the course of the algorithm) cannot see an error in the grid structure, but notice that the grid structure wraps around in some way, while type 1 pointers are for nodes that can actually see an error.
Here, with ``wrapping around", we mean that there is a node $v$ and a sequence of edge labels $L_1, L_2, \dots, L_p$ such that
\begin{enumerate}
\item there exists a dimension $j$ such that the number of labels $(\dirprev,j)$ in this sequence is different from the number of labels $(\dirnext,j)$, and
\item $z_v(L_1, L_2, \dots, L_p) =v$, i.e., we can walk from some node to itself without going in each dimension the same number of times in one direction as in the other.
\end{enumerate}
If the grid structure wraps around, then there must be an error somewhere (and nodes that see that the grid structure wraps around know where to point their error pointer to), or following an error pointer chain results in a cycle; however, since the constraints we put on error pointer chains are \emph{local} constraints (as we want to define an \lcl{} problem), the global behaviour of the chain is irrelevant.
We will not explicitly prove the global statements made in this informal overview; for our purposes it is sufficient to focus on the local views of nodes.

Note that if a chain of type 0 error pointers does not cycle, then at some point it will turn into a chain of type 1 error pointers, which in turn will end in an error.
Chains of type 1 error pointers cannot cycle.
We refer to Figure~\ref{fig:errorPointer} for an example of an error pointer chain.

\begin{figure}
	\centering
	\includegraphics[width=0.35\textwidth]{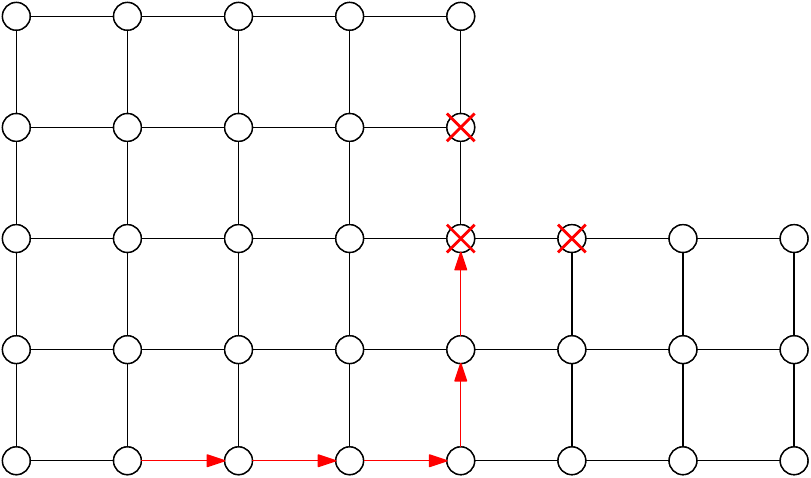}
	\caption{An example of an error pointer chain (shown in red). Nodes that are marked with a red cross are those who actually see an error in the grid structure. The output of only some of the depicted nodes is shown.}
	\label{fig:errorPointer}
\end{figure}

An output labelling for problem $\Pi_M^i$ is correct if the following conditions are satisfied.
\begin{enumerate}
	\item Each node $v$ produces at least one output label. If $v$ produces at least two output labels, then all of $v$'s output labels are contained in $\lbalabels \setminus \{ \lfinished \}$.
	\item Each node at which the input labelling does not satisfy the local grid graph constraints given in Section~\ref{app:local-checkability} outputs \error{}. All other nodes do not output \error{}.
	\item  If a node $v$ outputs $\lexemptu$ or $\lfinished$, then $v$ has at least one incident edge $e$ with input label $L_v(e)=(\dirprev,j)$, where $j\in \{1, \dots, i\}$.
	\item If the output labels of a node $v$ are contained in $\lbalabels \setminus \{ \lfinished \}$, then either there is a node in $v$'s $2$-radius neighbourhood that outputs a label from \lseterror, or the output labels of all nodes in $v$'s $2$-radius neighbourhood are contained in $\lbalabels$. Moreover, in the latter case $v$'s $2$-radius neighbourhood has a valid grid structure and the local constraints of an $M$-encoding graph, given in Section~\ref{subsec:locallba}, are satisfied at $v$.
	\item If the output of a node $v$ is $\lunbalanced{}$, then either there is a node in $v$'s $i$-radius neighbourhood that outputs a label from \lseterror, or the output labels of all nodes in $v$'s $i$-radius neighbourhood are contained in \lsetunbalanced. Moreover, in the latter case $v$'s $i$-radius neighbourhood has a valid grid structure and the local constraints for a proof of unbalance, given in Section~\ref{subsec:unbalanced}, are satisfied at $v$.
	\item Let $v$ be a node that outputs an error pointer $(s, r)$. Then $z_v(s)$ is defined, i.e., there is exactly one edge incident to $v$ with input label $s$. Let $u$ be the neighbour reached by following this edge from $v$, i.e., $u = z_v(s)$. Then $u$ outputs either $\error$ or an error pointer $(s', r')$, where in the latter case the following hold:
	\begin{itemize}
		\item $r' \ge r$, i.e., the type of the pointer cannot decrease when following a chain of error pointers;
		\item if $r' = 0 = r$, then $s' = s$, i.e., the pointers in a chain of error pointers of type 0 are consistently oriented;
		\item if $r' = 1 = r$ and
		\begin{align*}
		s &\in \bigl\{ (\dirprev,j),\, (\dirnext,j) \bigr\}, \\
		s' &\in \bigl\{ (\dirprev,{j'}),\, (\dirnext,{j'}) \bigr\},
		\end{align*}
		then $j' \ge j$, i.e., when following a chain of error pointers of type 1, the dimension of the pointer cannot decrease;
		\item if $r' = 1 = r$ and
		\[s, s' \in \bigl\{ (\dirprev,j),\, (\dirnext,j) \bigr\}\]
		for some $1 \le j \le i$, then $s' = s$, i.e., any two subsequent pointers in the same dimension have the same direction.
	\end{itemize}
\end{enumerate}
These conditions are clearly locally checkable, so $\Pi_M^i$ is a valid \lcl{} problem.

\subsection{Time complexity}\label{ssec:timecomplexity}
Let $M$ be an \shortname{}, $i \geq 2$ an integer, and $B$ the smallest positive integer satisfying $n \le B^{i-1} \cdot T_M(B)$. We will only consider \shortname{}s $M$ with the property that $B \le T_M(B)$ and for any two tape sizes $B_1 \ge B_2$ we have $T_M(B_1) \ge T_M(B_2)$.
In the following, we prove that $\Pi_M^i$ has time complexity $\Theta(n / B^{i-1}) = \Theta(T_M(B))$.
\subsubsection{Upper bound}\label{sssec:upperbound}
In order to show that $\Pi_M^i$ can be solved in $O(T_M(B))$ rounds, we provide an algorithm $\A$ for $\Pi_M^i$. Subsequently, we prove its correctness and that its running time is indeed $O(T_M(B))$. Algorithm $\A$ proceeds as follows.

First, each node $v$ gathers its constant-radius neighbourhood, and checks whether there is a local error in the grid structure at $v$, i.e., if constraints given in Section~\ref{app:local-checkability} are not satisfied. In that case, $v$ outputs $\error$. 
Then, each node $v$ that did not output $\error$ gathers its $R$-radius neighbourhood, where $R = c \cdot T_M(B)$ for a large enough constant $c \ge i$, and acts according to the following rules.
\begin{itemize}
	\item If there is a node labelled \error{} in $v$'s $R$-radius neighbourhood, then $v$ outputs an error pointer $(s, 1)$ of type 1, where $s \in \{ (\dirprev,j), (\dirnext,j) \}$ has the following property: among all shortest paths from $v$ to some node that outputs $\error$, there is one where the first edge $e$ on the path has input label $L_v(e) = s$, but, for any $j' < j$, there is none where the first edge $e$ has input label $L_v(e) \in \{ (\dirprev,{j'}), (\dirnext,{j'}) \}$.
	\item Now consider the case that there is no node labelled \error{} in $v$'s $R$-radius neighbourhood, but there is a path $P$ from $v$ to itself with the following property: Let $\mathcal L'$ be the sequence of labels read on the edges when traversing $P$, where for each edge $e = \{ u, w\}$ traversed from $u$ to $w$ we only read the label $L_u(e)$. Then there is some $1 \le j \le i$ such that the number of occurrences of label $(\dirprev,j)$ in $\mathcal L'$ is not the same as the number of occurrences of label $(\dirnext,j)$ in $\mathcal L'$. (In other words, the grid structure wraps around in some way.) Let $k$ be the smallest $j$ for which such a path $P$ exists. Then $v$ outputs an error pointer $(s, 0)$ of type 0, where $s = (\dirnext,{k})$.
	\item If the previous two cases do not apply (i.e., the input graph has a valid grid structure and does not wrap around, as far as $v$ can see), then $v$ checks for each dimension $1 \le j \le i$ whether in $v$'s $R$-radius neighbourhood there is both a node that does not have an incident edge labelled $(\dirprev,j)$ and a node that does not have an incident edge labelled $(\dirnext,j)$. (As we allow arbitrary input graphs, there could be several such node pairs.) For each dimension $j$ for which such two nodes exist, $v$ computes the size $d_j$ of the dimension by determining the distance between those two nodes w.r.t.\ dimension $j$, i.e., the absolute difference of the $j$-th coordinates of the two nodes. (Note that $v$ does not know the absolute coordinates, but can assign coordinates to the nodes it sees in a locally consistent manner, and that the absolute difference of the coordinates of those nodes does not depend on $v$'s choice as long as it is consistent.) Here, and in the following, $v$ assumes that the input graph also continues to be a grid graph outside of $v$'s $R$-radius neighbourhood. Then, $v$ checks whether among these $j$ there is a dimension $2 \le j' \le i$ with $d_{j'} \le T_M(B)$ that, in case $v$ actually computed the size of dimension $1$, also satisfies $d_{j'} \le d_1$. Now there are two cases:
	\begin{enumerate}
		\item If such a $j'$ exists, then $v$ chooses the smallest such $j'$ (breaking ties in a consistent manner), denoted by $k$, and computes its coordinate in dimension $k$. Node $v$ also computes its coordinate in dimension $1$ or verifies that it is larger than $T_M(B)$. Since $v$ can determine whether it has coordinate $0$ in all the other dimensions, it has all the information required to compute its output labels in the $M$-encoding graph where the execution of $M$ takes place on the surface that expands in dimensions $1$ and $k$. Consequently, $v$ outputs these labels (that is, labels from $\lbalabels$, defined in Section~\ref{subsec:machineencoding}). Note further that, according to the definition of an $M$-encoding graph, $v$ outputs $\lfinished$ if it verifies that its coordinate in dimension $1$ is larger than $T_M(B)$, even if it has coordinate $0$ in all dimensions except dimension $1$ and (possibly) $k$. Note that if the input graph does not continue to be a grid graph outside of $v$'s $R$-radius neighbourhood, then neighbours of $v$ might output error pointers, but this is still consistent with the local constraints of $\Pi_M^i$.
		\item If no such $j'$ exists, then, by the definition of $B$, node $v$ sees (nodes at) both borders of dimension $1$. In this case, $v$ can compute the label it would output in a proof of unbalance (that is, a label from \lsetunbalanced, defined in Section~\ref{subsec:unbalanced}), since for this, $v$ only has to determine whether its coordinates are the same in all dimensions (which is possible as all nodes with this property are in distance at most $i \cdot T_M(B)$ from the origin). Consequently, $v$ outputs this label. Again, if the input graph does not continue to be a grid graph outside of $v$'s $R$-radius neighbourhood, then, similarly to the previous case, the local constraints of $\Pi_M^i$ are still satisfied.
	\end{enumerate}
\end{itemize}

\begin{theorem}\label{thm:upper-bound}
	Problem $\Pi_M^i$ can be solved in $O(T_M(B))$ rounds.
\end{theorem}
\begin{proof}
	We will show that algorithm $\A$ solves problem $\Pi_M^i$ in $O(T_M(B))$ rounds. It is easy to see that the complexity of $\A$ is $O(T_M(B))$. We need to prove that it produces a valid output labelling for $\Pi_M^i$. For this, first consider the case that the input graph is a grid graph.
	Let $2 \le k \le i$ be the dimension with minimum size (apart, possibly, from the size of dimension $1$).
	If $d_k \le d_1$, then $d_k \le T_M(B)$, by the definition of $B$ and the assumption that $T_M(B) \ge B$.
	In this case, according to algorithm $\A$, the nodes output labels that turn the input graph into an $M$-encoding graph, thereby satisfying the local constraints of $\Pi_M^i$.
	If, on the other hand, $d_k > d_1$, then according to algorithm $\A$, the nodes output labels that constitute a valid proof for unbalanced grids, again ensuring that the local constraints of $\Pi_M^i$ are satisfied.
	
	If the input graph looks like a grid graph from the perspective of some node $v$ (but might not be a grid graph from a global perspective), then there are two possibilities: either the input graph also looks like a grid graph from the perspective of all nodes in $v$'s $2$-radius neighbourhood, in which case the above arguments ensure that the local constraints of $\Pi_M^i$ (regarding $M$-encoding labels, i.e., labels from $\lbalabels$) are satisfied at $v$, or some node in $v$'s $2$-radius neighbourhood notices that the input graph is not a grid graph, in which case it outputs an error pointer and thereby ensures the local correctness of $v$'s output.
	The same argument holds for the local constraints of $\Pi_M^i$ regarding labels for proving unbalance (instead of labels from $\lbalabels$), with the only difference that in this case we have to consider $v$'s $i$-radius neighbourhood (instead of $v$'s $2$-radius neighbourhood).
	
	What remains to show is that the constraints of $\Pi_M^i$ are satisfied at nodes $v$ that output $\error$ or an error pointer.
	If $v$ outputs $\error$ according to $\A$, then the constraints of $\Pi_M^i$ are clearly satisfied, hence assume that $v$ outputs an error pointer $(s,r)$.
	
	We first consider the case that $r=0$, i.e., $v$ outputs an error pointer of type $0$.
	In this case, according to the specifications of $\A$, there is no error in the grid structure in $v$'s $R$-radius neighbourhood.
	Let $u$ be the neighbour of $v$ the error pointer points to, i.e., the node reached by following the edge with label $s$ from $v$.
	Due to the valid grid structure around $v$, node $u$ is well-defined. 
	According to the specification of $\Pi_M^i$, we have to show that $u$ outputs an error pointer $(r', s')$ satisfying $r' = 1$ or $s' = s$.
	If there is a node in $u$'s $R$-radius neighbourhood that outputs $\error$, then $u$ outputs an error pointer of type 1, i.e., $r' = 1$.
	Thus, assume that there is no such node, which implies that the grid structure in $u$'s $R$-radius neighbourhood is valid as well.
	
	Consider a path from $v$ to itself inside $v$'s $R$-radius neighbourhood, and let $L_1, \dotsc, L_h$ be the sequence of edge labels read when traversing this path, where for each edge $e = \{ w, x \}$, we only consider the input label that belongs to the node from which the traversal of the edge starts, i.e., $L_w(e)$ if edge $e$ is traversed from $w$ to $x$.
	Then, due to the grid structure of $v$'s $R$-radius neighbourhood, there is such a path $P$ with the following property: for each $1 \le j \le i$, at most one of $(\dirprev,j)$ and $(\dirnext,j)$ is contained in the edge label sequence (as any two labels $(\dirprev,j)$ and $(\dirnext,j)$ ``cancel out''), and the edge label sequence (and thus the directions of the edges) is ordered non-decreasingly w.r.t. dimension, i.e., if $L_{h'} \in \{ (\dirprev,{j'}), (\dirnext,{j'}) \}$ and $L_{h''} \in \{ (\dirprev,{j''}), (\dirnext,{j''}) \}$ for some $1 \le h' \le h'' \le h$, then $j' \le j''$.
	Also, we can assume that the edge label $L_1$ of the first edge on $P$ is of the kind $(\dirnext,j)$ for some $j$ as we can reverse the direction of path $P$ and subsequently transform it into a path with the above properties by reordering the edge labels.
	Due to the specification of $\A$ regarding type 0 error pointer outputs and the above observations, we can assume that $L_1 = s$.
	
	Consider the path $P'$ obtained by starting at $u$ and following the edge label sequence $L_2, \dotsc, L_h, s$.
	Since $u = z_v(L_1) = z_v(s)$ and $v = z_v(L_1, \dotsc, L_h)$, we have that $u = z_u(L_2, \dotsc, L_h, s)$.
	Since $P$ is contained in the $R$-radius neighbourhood of $v$ (and $P$ has the nice structure outlined above), $P'$ is contained in the $R$-radius neighbourhood of $u$, thereby ensuring that $u$ outputs a type 0 error pointer.
	Let $k$ and $k'$ be the indices satisfying $(\dirnext,k) = s$ and $(\dirnext,{k'}) = s'$, respectively.
	Again due to the specification of $\A$ regarding type 0 error pointer outputs, we see that $k' \le k$.
	However, using symmetric arguments to the ones provided above, it is also true that for each path $P''$ from $u$ to itself of the kind specified above, there is a path from $v$ to itself that contains the same labels in the label sequence as $P''$ (although not necessarily in the same order), which implies that $k \le k'$.
	Hence, $k' = k$, and we obtain $s' = s$, as required.
	
	Now consider the last remaining case, i.e., that $v$ outputs an error pointer $(s,1)$ of type~$1$.
	Again, let $u$ be the neighbour of $v$ the error pointer points to, i.e., the node reached by following the edge with label $s$ from $v$.
	Let $D$ and $D'$ be the lengths of the shortest paths from $v$, resp.\ $u$ to some node that outputs $\error$.
	By the specification of $\A$ regarding type 1 error pointer outputs, we know that $D' = D - 1$, which ensures that $u$ outputs $\error$ or an error pointer of type 1.
	If $u$ outputs $\error$, then the local constraints of $\Pi_M^i$ are clearly satisfied at $v$.
	Thus, consider the case that $u$ outputs an error pointer $(s', 1)$ of type 1.
	Let $k$ and $k'$ be the indices satisfying $s \in \{ (\dirprev,{k}), (\dirnext,{k}) \}$ and $s' \in \{ (\dirprev,{k'}), (\dirnext,{k'}) \}$, respectively.
	We need to show that either $k'=k$ and $s'=s$, or $k' > k$.
	
	Suppose for a contradiction that either $k' = k$ and $s '\neq s$, or $k' < k$.
	Note that the latter case also implies $s' \neq s$.
	Consider a path $P'$ of length $D'$ from $u$ to some node $w$ outputting $\error$ with the property that the first edge $e'$ on $P'$ has input label $s'$.
	Such a path $P'$ exists by the specification of $\A'$.
	Let $P$ be the path from $v$ to $w$ obtained by appending $P'$ to the path from $v$ to $u$ consisting of edge $e = \{ v, u \}$.
	Note that $L_v(e) = s$.
	Since $v$ did not output $\error$, the local grid graph constraints, given in Section~\ref{app:local-checkability}, are satisfied at $v$.
	Hence, if $k' < k$, we can obtain a path $P''$ from $v$ to $w$ by exchanging the directions of the first two edges of $P$, i.e., $P''$ is obtained from $P$ by replacing the first two edges $e, e'$ by the edges $e'' = \{ v, z_v(s') \}, e''' = \{ z_v(s'), z_u(s') \}$.
	Note that $L_v(e'') = s'$ and $L_{z_v(s')}(e''') = s$.
	In this case, since $P''$ has length $D$ and starts with an edge labelled $s'$, we obtain a contradiction to the specification of $\A$ regarding error pointers of type 1, by the definitions of $k, k', s, s', D$. 
	Thus assume that $k' = k$ and $s '\neq s$.
	In this case, $z_v(s, s') = z_u(s') = v$ which implies that $D' = D + 1$, by the definitions of $D, D', P'$.
	This is a contradiction to the equation $D' = D - 1$ observed above.
	Hence, the local constraints of $\Pi_M^i$ are satisfied at $v$.
\end{proof}

\subsubsection{Lower bound}
\begin{restatable}{theorem}{thmlowerbound}\label{thm:lower-bound}
	Problem $\Pi_M^i$ cannot be solved in $o(T_M(B))$ rounds.
\end{restatable}
\begin{proof}
	Consider $i$-dimensional grid graphs where the number $n$ of nodes satisfies $n =  B^{i-1} \cdot T_M(B)$.
	Clearly, there are infinitely many $n$ with this property, due to the definition of $B$.
	More specifically, consider such a grid graph $\G$ satisfying $d_j = B$ for all $j \in \{2,\ldots,i\}$, and $d_1 = T_M(B)$.
	By the local constraints of $\Pi_M^i$, the only valid global output is to produce an $M$-encoding graph, on a surface expanding in dimensions $1$ and $k$ for some $k \in \{2,\ldots,i\}$. In fact:
	\begin{itemize}
		\item If nodes try to prove that the grid graph is unbalanced, since $T_M(B) \ge B$, the proof must either be locally wrong, or, if nodes outputting \lunbalanced{} actually form a diagonal chain, this chain must terminate on a node that, for any $2 \le j \le i$, does not have an incident edge labelled $(\dirnext,{j})$, that is, constraints defined in Section~\ref{subsec:unbalanced} are not satisfied, which also violates the local constraints of $\Pi_M^i$.
		\item If nodes try to produce an error pointer, since the specification of the validity of pointer outputs in the local constraints of $\Pi_M^i$ ensures that on grid graphs a pointer chain cannot visit any node twice, any error pointer chain must terminate somewhere. Since no nodes can be labelled \error{}, this is not valid.
		\item The only remaining possibility for the origin is to output a label from $\lbalabels \setminus \{ \lfinished \}$, which already implies that all the other nodes must produce outputs that turn the graph into an $M$-encoding graph.
	\end{itemize}
	Thus, it remains to show that producing a valid $M$-encoding labelling requires time $\Omega(T_M(B))$. Consider the node having coordinate $1$ equal to $x=T_M(B)$ and all other coordinates equal to $0$. This node must be labelled $\lstate({f})$, the nodes with coordinate $1$ strictly less than $x$ must not be labelled $\lstate({f})$, and the nodes with coordinate $1$ strictly greater than $x$ must be labelled \lfinished{}. Thus, a node needs to know if it is at distance $T_M(B)$ from the boundary of coordinate $1$, which requires $\Omega(T_M(B))$ time.
\end{proof}

\subsection{\boldmath Instantiating the \texorpdfstring{\lcl{}}{LCL} construction}\label{subsec:instantiate}
Our construction is quite general and allows to use a wide variety of \shortname{}s to obtain many different \lcl{} complexities. As a proof of concept, in Theorems \ref{thm:firstcompl} and \ref{thm:secondcompl}, we show some complexities that can be obtained using some specific \shortname{}s. Recall that
\begin{itemize}
	\item if we choose our \shortname{} $M$ to be a unary $k$-counter, for constant $k$, then $M$ has a running time of $T_M(B) = \Theta(B^k)$, and
	\item if we choose $M$ to be a binary counter, then $M$ has a running time of $T_M(B) = \Theta(2^B)$.
\end{itemize}

\begin{theorem}\label{thm:firstcompl}
	For any rational number $0 \le \alpha \le 1$, there exists an \lcl{} problem with time complexity $\Theta(n^{\alpha})$.
\end{theorem}
\begin{proof}
	Let $j > k$ be positive integers satisfying $\alpha = k/j$. 
	Given an \shortname{} $M$ with running time $\Theta(B^k)$ and choosing $i = j - k + 1$, we obtain an \lcl{} problem $\Pi_M^i$ with complexity $\Theta(n / B^{j - k})$. We have that $n = \Theta(B^{j-k} \cdot T_M(B)) = \Theta(B^{j})$, which implies $B = \Theta(n^{1/j})$. Thus the time complexity of $\Pi_M^i$ is $\Theta( n / n^{(j-k)/j}) = \Theta( n^{\alpha})$.
\end{proof}

\begin{theorem}\label{thm:secondcompl}
	There exist \lcl{} problems of complexities $\Theta(n/\log^{j} n)$, for any positive integer~$j$.
\end{theorem}
\begin{proof}
	Given an \shortname{} $M$ with running time $\Theta(2^B)$ and choosing $i = j+1$, we obtain an \lcl{} problem $\Pi_M^i$ with complexity $\Theta(n / B^j)$. We have that $n = \Theta(B^j \cdot T_M(B)) = \Theta(B^j \cdot 2^B)$, which implies $B = \Theta(\log n)$. Thus the time complexity of $\Pi_M^i$ is $\Theta( n / \log^j n)$.
\end{proof}

\section{Complexity gap on trees}\label{sec:trees}

In the previous section, we have seen that there are infinite families of \lcl{}s with distinct time complexities between $\omega(\sqrt{n})$ and $o(n)$. In this section we prove that on trees there are no such \lcl{}s. That is, we show that if an \lcl{} is solvable in $o(n)$ rounds on trees, it can be also solved in $O(\sqrt{n})$ rounds.

The high level idea is the following. Consider an \lcl{} $\Pi$ that can be solved in $o(n)$ rounds on a tree $T$ of $n$ nodes, that is, there exists a distributed algorithm $\aorig$, that, running on each node of $T$, outputs a valid labelling for $\Pi$ in sublinear time. We show how to speed this algorithm up, and obtain a new algorithm $\aorig'$ that runs in $O(\sqrt{n})$ rounds and solves $\Pi$ as well.

To do this, we show that nodes of $T$ can distributedly construct, in $O(\sqrt{n})$ rounds, a virtual graph $S$ of size $N \gg n$. This graph $S$ will be defined such that we can run $\aorig$ on it, and use the solution that we get to obtain a solution for $\Pi$ on $T$. Moreover, for each node of $T$, it will be possible to simulate the execution of $\aorig$ on $S$ by just inspecting its neighbourhood of radius $O(\sqrt{n})$ on $T$, thus obtaining an algorithm for $\Pi$ running in $O(\sqrt{n})$ rounds.

We will define $S$ in multiple steps. Intuitively, $S$ is defined by first pruning branches of $T$ of small radius, and then \emph{pumping} (in the theory of formal languages sense) long paths to make them even longer. In more detail, in Section~\ref{sub:skeleton} we will define the concept of a \emph{skeleton tree}, where, starting from a tree $T$ we define a tree $T'$ where all subtrees of $T$ having a height that is less than some threshold are removed. Then, in Section~\ref{sub:virtual}, we will prune $T'$ even more, by removing all nodes of $T'$ having degree strictly greater than $2$, obtaining forest $T''$. This new forest will be a collection of paths. We will then split these paths in shorter paths and pump each of them. Intuitively, the pump procedure will replace the middle part of these paths by a repeated pattern that depends on the original content of the paths and on the parts previously removed when going from $T$ to $T'$. Tree $S$ will be obtained starting from the result of the pumping procedure, by bringing back the parts removed when going from $T$ to $T'$. Also, throughout the definition of $S$ we will keep track of a (partial) mapping between nodes of $T$ and nodes of $S$. 

Then, in  Section~\ref{sub:properties} we will prove useful properties of $S$. One crucial property, shown in Lemma \ref{lem:radiuslowerbound}, will be that, if two nodes in $T$ are far enough (that is, at $\omega(\sqrt{n})$ distance), then their corresponding nodes of $S$ will be at much larger distance. In Section~\ref{sub:faster} we will use this property to show that we can execute $\aorig$ on $S$ by inspecting only a neighbourhood of radius $O(\sqrt{n})$ of $T$. Notice that we will have conflicting requirements. On one hand, by pumping enough, the size of the graph increases to $N \gg n$, and $\aorig$ on $S$ will be allowed to run for $t = o(N)$ rounds, that is, much more than the time allowed on $T$. This seems to give us the opposite effect of what we want, that is, we actually increased the running time instead of reducing it. On the other hand, we will prove that seeing at distance $t$ on $S$ requires to see only at distance $O(\sqrt{n})$ on $T$, hence we will effectively be able to run $\aorig$ within the required time bound. We will use the output given by $\aorig$ on $S$ to obtain a partial solution for $T$, that is, only some nodes of $T$ will fix their output using the same output of their corresponding node in $S$. There will be some nodes that remain unlabelled: those nodes that correspond to the pumped regions of $S$. Finally, in Section~\ref{sub:filling}, we will show that it is possible to complete the unlabelled regions of $T$ efficiently in a valid manner, heavily using techniques already presented in \cite{Chang2019}.

\subsection{Skeleton tree}\label{sub:skeleton}
We first describe how, starting from a tree $T=(V,E)$, nodes can distributedly construct a virtual tree $T'$, called the \emph{skeleton} of $T$. Intuitively, $T'$ is obtained by removing all subtrees of $T$ having a height that is less than some threshold $\tnew$.

More formally, let $\tnew = \tc \sqrt{n}$, for some constant $\tc$ that will be fixed later.  Each node $v$ starts by gathering its $\tnew$-radius neighbourhood, $\ball{v}$. Also, let $\dg{v}$ be the degree of node $v$ in $T$. For all $v \in V$, we partition nodes of $\ball{v}$ (excluding $v$) in $\dg{v}$ components (one for each neighbour of $v$). Let us denote these components with $C_i(v)$, where $1\le i\le \dg{v}$. Each component $C_i(v)$ contains all nodes of $\ball{v}$ present in the subtree rooted at the $i$-th neighbour of $v$, excluding $v$.

Then, each node marks as $\del$ all the components that have low depth and broadcasts this information. Informally, nodes build the skeleton tree by removing all the components that are marked as $\del$ by at least one node. More precisely, each node $v$, for each $C_i(v)$, if $\dist{v}{w} < \tnew$ for all $w$ in $V(C_i(v))$, marks all edges in $E(C_i(v)) \cup \{\{v,u\}\}$ as $\del$, where $u$ is the $i$-th neighbour of $v$. Then, $v$ broadcasts $\ball{v}$ and the edges marked as $\del$ to all nodes at distance at most $\tnew + 2\tc$. Finally, when a node $v$ receives messages containing edges that have been marked with $\del$ by some node, then also $v$ internally marks as $\del$ those edges.

Now we have all the ingredients to formally describe how we construct the skeleton tree. The skeleton tree $T' = (V',E')$ is defined in the following way. Intuitively, we keep only edges that have not been marked $\del$, and nodes with at least one remaining edge (i.e., nodes that have at least one incident edge not marked with $\del$). In particular,
\begin{align*}
E' &= \bigl\{e\in E(T) \bigm| e \text{ is not marked with } \del\bigr\},\\
V' &= \bigl\{ u\in V \bigm| \text{there is a } w\in V \text{ s.t. } \{u,w\} \in E' \bigr\}.
\end{align*}
Also, we want to keep track of the mapping from a node of $T'$ to its original node in $T$; let $\phi \colon V(T') \to V(T)$ be such a mapping. Finally, we want to keep track of deleted subtrees, so let $\deltree_v$ be the subtree of $T$ rooted at $v \in V'$ containing all nodes of $C_j(v)$, for all $j$ such that $C_j(v) \text{ has been marked as } \del$. See Figure~\ref{fig:T_to_T'} for an example.

\begin{figure*}
	\centering
	\includegraphics[width=\figscale\textwidth]{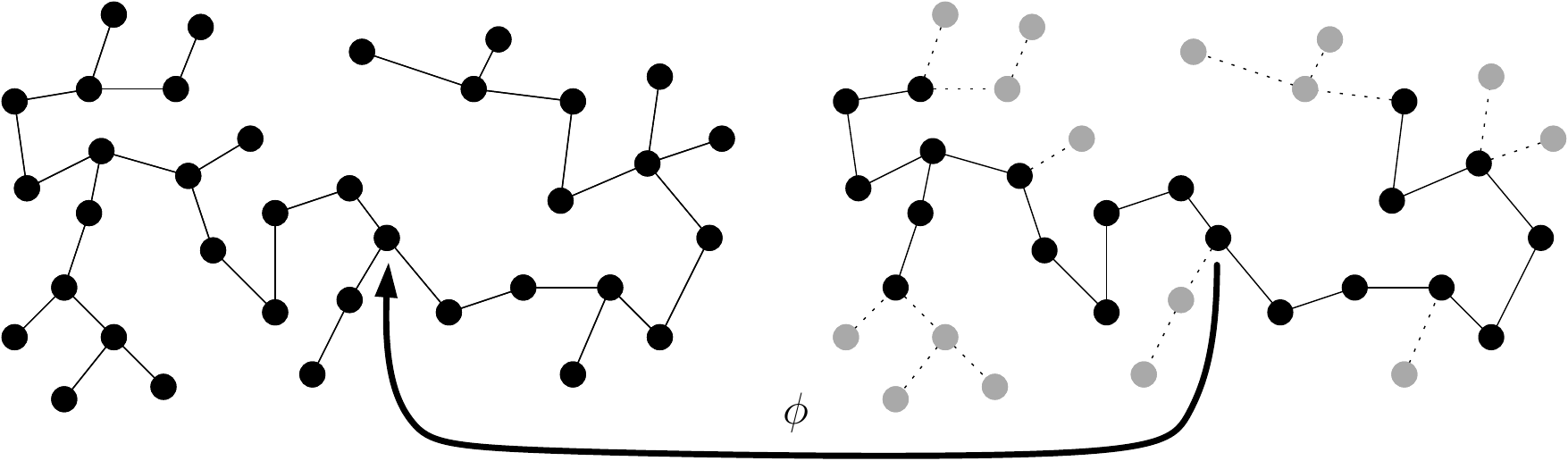}
	\caption{An example of a tree $T$ and its skeleton $T'$; nodes removed from T in order to obtain $T'$ are shown in grey. In this example, $\tnew$ is $3$.}
	\label{fig:T_to_T'}
\end{figure*}

\subsection{Virtual tree}\label{sub:virtual}
We now show how to distributedly construct a new virtual tree, starting from $T'$, that satisfies some useful properties. The high level idea is the following. The new tree is obtained by \emph{pumping} all paths contained in $T'$ having length above some threshold. More precisely, by considering only degree-$2$ nodes of $T'$ we obtain a set of paths. We split these paths in shorter paths of length $l$ ($\tc \le l \le 2 \tc$) by computing a  $(\tc +1,\tc)$ ruling set. Then, we pump these paths in order to obtain the final tree. Recall a $(\alpha,\beta)$ ruling set $R$ of a graph $G$ guarantees that nodes in $R$ have distance at least $\alpha$, while nodes outside $R$ have at least one node in $R$ at distance at most $\beta$.  It can be distributedly computed in $O(\log^* n)$ rounds using standard colouring algorithms \cite{Linial1992}.

More formally, we start by splitting the tree in many paths of short length. Let $T''$ be the forest obtained by removing from $T'$ each node $v$ having $\dgg{v}{T'} > 2$ (that is, the degree of $v$ in $T'$). $T''$ is a collection $\paths$ of disjoint paths. Let $\psi \colon V(T'') \to V(T')$ be the mapping from nodes of $T''$ to their corresponding node of $T'$. See Figure~\ref{fig:T'_to_T''} for an example.

\begin{figure*}
	\centering
	\includegraphics[width=\figscale\textwidth]{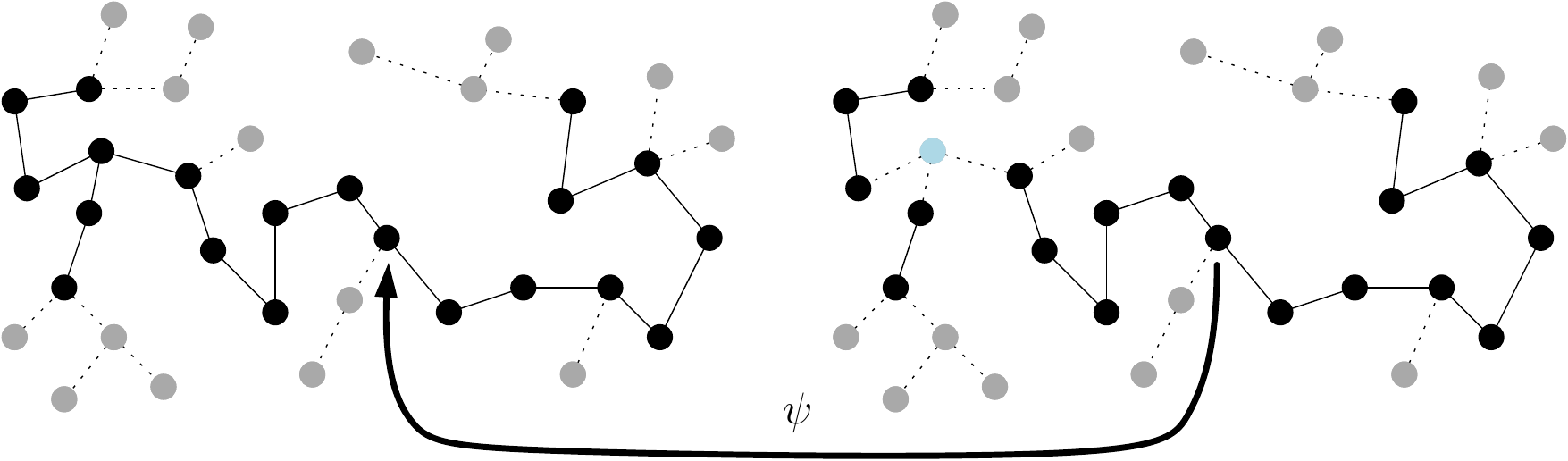}
	\caption{An example of the tree $T''$ obtained from $T'$; nodes with degree greater than $2$ (in blue) are removed from $T'$.}
	\label{fig:T'_to_T''}
\end{figure*}

We now want to split long paths of $\paths$ in shorter paths. In order to achieve this, nodes of the same path can efficiently find a $(\tc +1,\tc)$ ruling set in the path containing them. Nodes not in the ruling set form short paths of length $l$, such that $\tc \le l \le 2\tc$, except for some paths of $\paths$ that were already too short, or subpaths at the two ends of a longer path (this can happen when a ruling set node happens to be very near to the endpoint of a path of $\paths$).
Let $\spaths$ be the subset of the resulting paths having length $l$ satisfying $\tc \le l \le 2\tc$. See Figure~\ref{fig:T''_to_Q} for an example.

\begin{figure*}
	\centering
	\includegraphics[width=\figscale\textwidth]{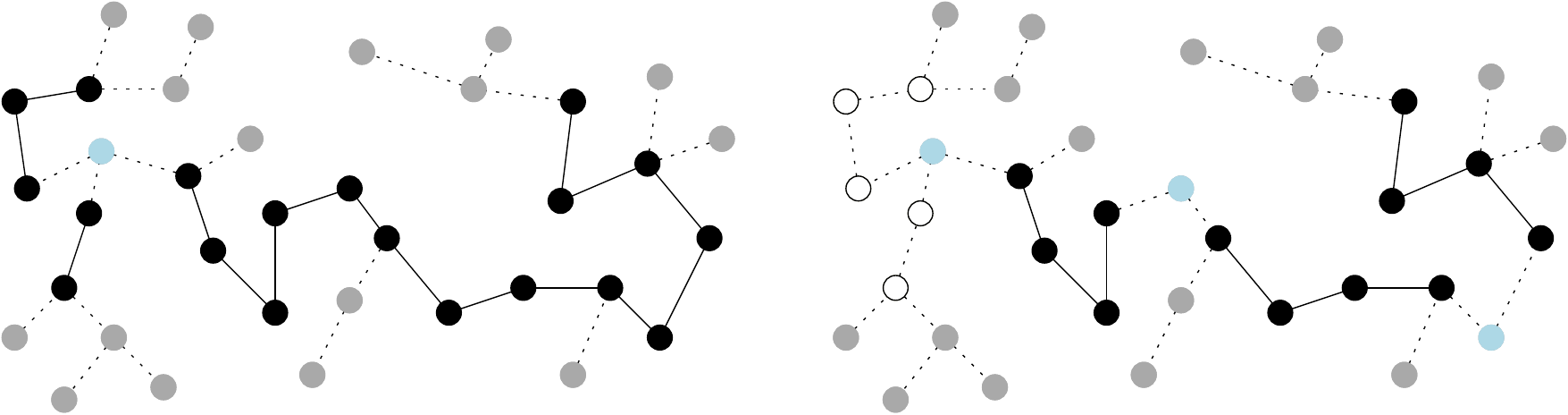}
	\caption{Blue nodes break the long paths $\paths$ of $T''$ shown on the left into short paths $\spaths$ shown in black on the right; short paths (in the example, paths with length less then 4) are ignored.}
	\label{fig:T''_to_Q}
\end{figure*}

In order to obtain the final tree, we will replace paths in $\spaths$ with longer version of them. We will first describe a function, $\replace$, that can be used to replace a subgraph with a different one. Informally, given a graph $G$ and a subgraph $H$ connected to the other nodes of $G$ via a set of nodes $F$ called poles, and given another graph $H'$, it replaces $H$ with $H'$. This function is a simplified version of the function $\replace$ presented in \cite[Section 3.3]{Chang2019}.
\begin{definition}[$\replace$]
	Let $H$ be a subgraph of $G$, and let $H'$ be an arbitrary graph. The poles of $H$ are those vertices in $V(H)$ adjacent to some vertex in $V(G)\setminus V(H)$. Let $F=(v_1,\ldots,v_p)$ be a list of the poles of $H$, and let $F'=(v'_1,\ldots,v'_p)$ be a list of nodes contained in $H'$ (called poles of $H'$). The graph $G' = \replace(G,(H,F),(H',F'))$ is defined in the following way. Beginning with $G$, replace $H$ with $H'$, and replace any edge $\{u,v_i\}$, where $u \in V(G)\setminus V(H)$, with $\{u,v'_i\}$.
\end{definition}
Informally, we will use the function $\replace$ to substitute each path $Q \in \spaths$ with a longer version of it, that satisfies some useful properties. In Section~\ref{sub:filling} we will have enough ingredients to be able to define a function, $\pump$, that is used to obtain these longer paths. This function will be defined in an analogous way of the function $\pump$ presented in \cite[Section 3.8]{Chang2019}. For now, we just define some properties that this function must satisfy.
\begin{definition}[Properties of $\pump$]
	Given a path $Q \in \spaths$ of length $l$ ($\tc \le l \le 2\tc$), consider the subgraph $Q^T$ of $T$, containing, for each $v \in V(Q)$, the tree $\deltree_{\chi(v)}$ (recall that $\deltree_v$ is the tree rooted at $v$ containing nodes removed when pruning $T$, defined at the end of Section~\ref{sub:skeleton}), where  $\chi(v) = \phi(\psi(v)))$, that is, the path $Q$ augmented with all the nodes deleted from the original tree that are connected to nodes of the path. Let $v_1,v_2$ be the endpoints of $Q$.
	
	The function $\pump(Q^T,\pumpmult)$ produces a new tree $P^T$ having two nodes, $v'_1$ and $v'_2$, satisfying that the path between $v'_1$ and $v'_2$ has length $l'$, such that $\tc \pumpmult \le l' \le \tc (\pumpmult+1)$. The new tree is obtained by replacing a subpath of $Q$, along with the deleted nodes connected to it, with many copies of the replaced part, concatenated one after the other.
	$\pump$ satisfies that nodes $v'_1,v'_2 \in G'$, where $G' = \replace(G,(Q^T,(v_1,v_2)),(P^T,(v'_1,v'_2)))$, have the same view as $v_1,v_2 \in G$ at distance  $2 \checkradius$ (where $\checkradius$ is the \lcl{} checkability radius). Note that, in the formal definition of $\pump$, we will set $c$ as a function of $r$.
\end{definition}
Let $\spaths^T$ be the set containing all $Q^T$. See Figure~\ref{fig:Q_to_Qt} for an example of $\spaths^T$.

\begin{figure*}
	\centering
	\includegraphics[width=\figscale\textwidth]{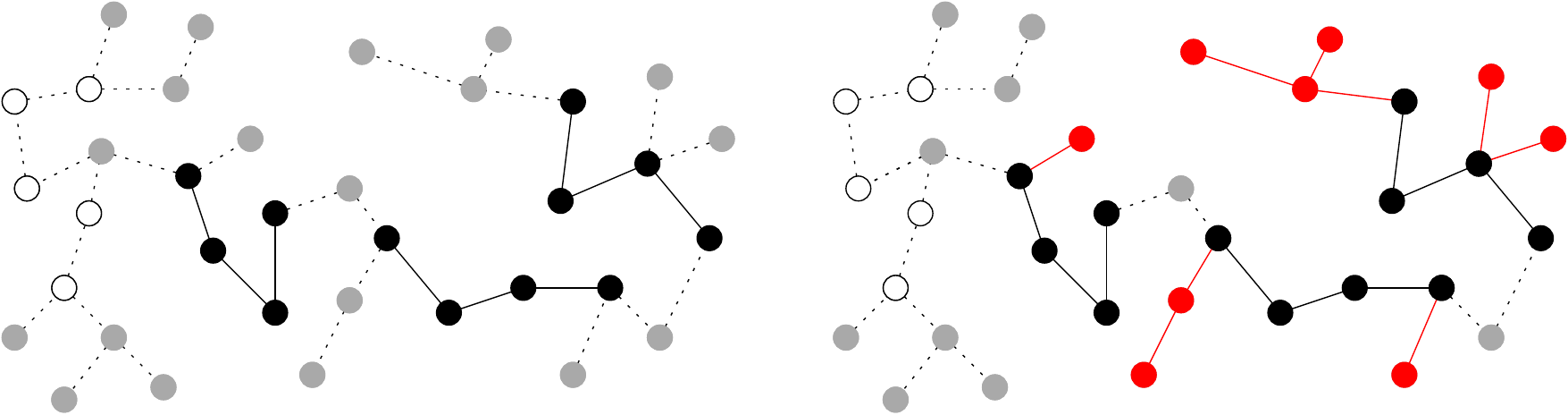}
	\caption{An example of $\spaths^T$, obtained by merging the path nodes (in black) with previously removed trees connected to them (in red).}
	\label{fig:Q_to_Qt}
\end{figure*}

The final tree $S$ is obtained from $T$ by replacing each path $Q \in \spaths$ in the following way. Replace each subgraph $Q^T$ with $P^T = \pump(Q^T,\pumpmult)$. Note that a node $v$ cannot see the whole set $\spaths$, but just all the paths $Q \in \spaths$ that end at distance at most $\tnew + 2 \tc$ from $v$. Thus each node locally computes just a part of $S$, that is enough for our purpose. We call the subgraph of $Q^T$ induced by the nodes of $Q$ the \emph{main path} of $Q^T$, and we define the main path of $P^T$ in an analogous way. See Figure~\ref{fig:Qt_to_S} for an example.

\begin{figure*}
	\centering
	\includegraphics[width=\figscale\textwidth]{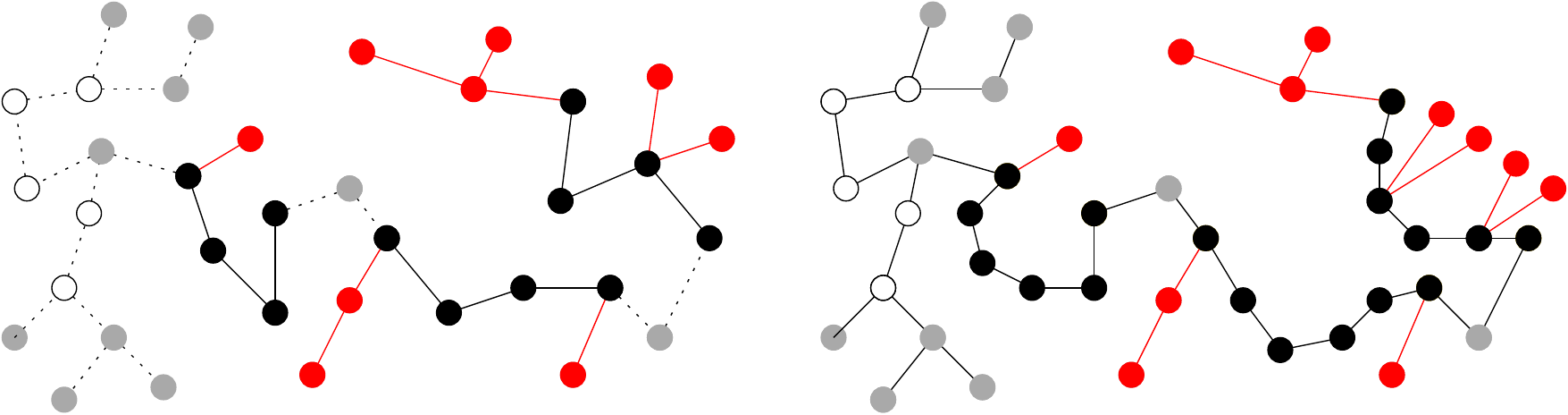}
	\caption{Tree $S$ (on the right) is obtained by pumping the black paths.}
	\label{fig:Qt_to_S}
\end{figure*}

Finally, we want to keep track of the \emph{real} nodes of $S$, that will be nodes that have not been removed when creating the skeleton tree $T'$ and are also not part of the pumped regions. Nodes of $S$ are divided in two parts, $S_o$ and $S_p$. The set $S_o$ contains all nodes of $T'$ that are not contained in any $Q^T$, and all nodes that are at distance at most $2 \checkradius$ from nodes not contained in any $Q^T$, while $S_p = V(S) \setminus S_o$. Let $\eta$ be a mapping from real nodes of the virtual graph ($S_o$) to their corresponding node of $T$ (this is well defined, by the properties of $\pump$), and let $T_o = \{ \eta(v) ~|~  v \in S_o\}$ (note that also $\eta^{-1}$ is well defined for nodes in $T_o$). Informally, $T_o$ is the subset of nodes of $T$ that are far enough from pumped regions of $S$, and have not been removed while creating $T'$. Note that we use the function $\eta$ to distinguish between nodes of $S$ and nodes of $T$, but $\eta$ is actually the identity function between a subset of shared nodes. This concludes the definition of $S$, as a function of the original tree $T$, and two parameters, $B$ and $c$. Let $\virt$ be the function that maps $T$ to $S$, that is, $S = \virt(T,B,\tc)$. See Figure~\ref{fig:S_to_T} for an example.

\begin{figure*}
	\centering
	\includegraphics[width=\figscale\textwidth]{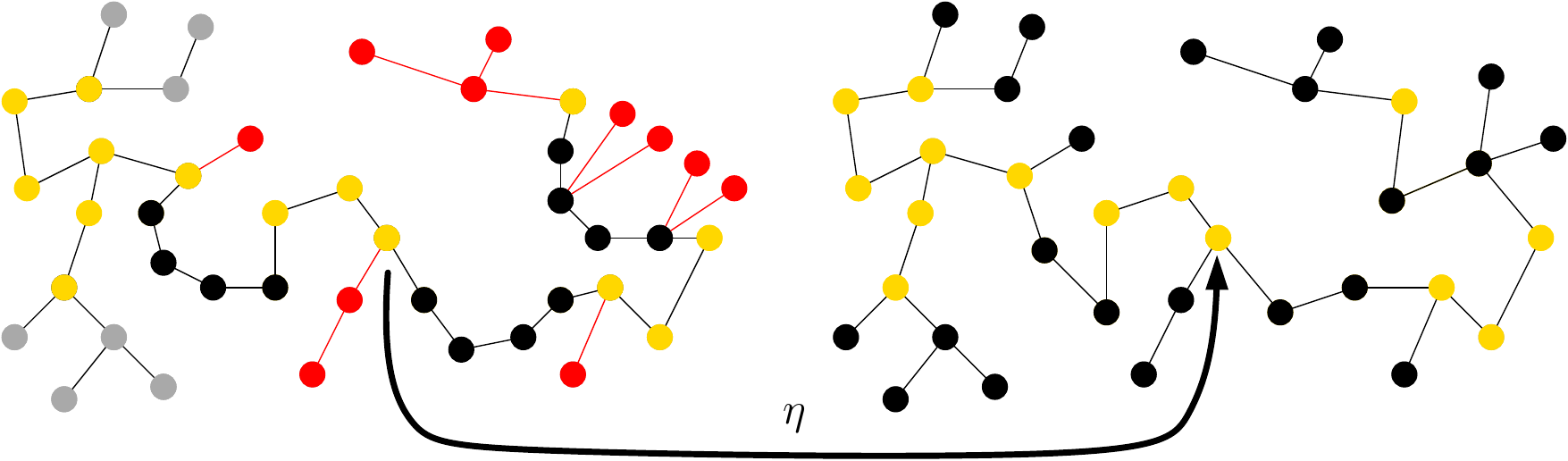}
	\caption{Nodes in yellow on the left are the ones in $S_o$, while the yellow ones on the right are nodes in $T_o$. Note that, for the sake of simplicity, we consider $2r=1$.}
	\label{fig:S_to_T}
\end{figure*}

\subsection{Properties of the virtual tree}\label{sub:properties}
We will now prove three properties about the virtual graph $S$. The first one provides an upper bound on the number of nodes of $S$, as a function of the number of nodes of $T$. This will be useful when executing $\aorig$ on $S$. In that case, we will lie to $\aorig$ about the size of the graph, by telling to the algorithm that there are $N=\tc (\pumpmult+1) n$ nodes. This lemma will guarantee us that the algorithm cannot see more than $N$ nodes and notice an inconsistency.

\begin{lemma}
	\label{lem:newsizeupperbound}
	The tree $S$ has at most $N=\tc (\pumpmult+1) n$ nodes, where $n = |V(T)|$, and $S = \virt(T,\pumpmult,\tc)$.
\end{lemma}
\begin{proof}
	$S$ is obtained by pumping $T$.
	The main path of the subtree obtained by pumping some $Q^T \in \spaths^T$ has length at most $\tc (\pumpmult+1)$. This implies that each node of the main path of $Q^T$ is copied at most $\tc (\pumpmult+1)$ times. Also, a deleted tree $\deltree_v$ rooted at some path node $v$ is not connected to more than one path node. Thus, \emph{all} nodes of $T$ are copied at most $\tc (\pumpmult+1)$ times.
\end{proof}

The following lemma bounds the size of $T'$ compared to the size of $T''$. This bound will be useful in the proof of Lemma \ref{lem:radiuslowerbound}. Notice that, this is the exact point in which our approach stops working for time complexities of $O(\sqrt{n})$ rounds. This is exactly what we expect, since we know that there are \lcl{} problems on trees having complexity $\Theta(\sqrt{n})$~\cite{Chang2019}.
\begin{lemma}
	\label{lem:deg2nodes}
	For any path $P = (x_1,\ldots,x_k)$ of length $k \ge \tc \sqrt{n}$ that is a subgraph of $T'$, at most $\frac{\sqrt{n}}{\tc}$ nodes in $V(P)$ have degree greater than $2$. 
\end{lemma}
\begin{proof}
	If a node $x_j \in P$ has $\dgg{v}{T'} > 2$, it means that it has at least one neighbour $z \not\in \{x_{j-1},x_{j+1}\}$ in $T'$ such that there exists a node $w$ satisfying $\dist{x_j}{w} \ge \tnew$ such that the shortest path connecting $x_j$ and $w$ contains $z$. Thus, for each node in $P$ with $\dgg{v}{T'} > 2$, we have at least other $\tnew$ nodes not in $P$. If at least $\frac{\sqrt{n}}{\tc}+1$ nodes of $P$ have degree greater than~$2$, we would obtain a total of  $(\frac{\sqrt{n}}{\tc}+1)\cdot \tnew > n$ nodes, a contradiction. 
\end{proof}

The following lemma compares distances in $T$ with distances in $S$, and states that if two nodes are far enough, that is, at $\omega(\sqrt{n})$ distance in $T$, then we can increase the distance of their corresponding nodes in $S$ by an arbitrary amount. This is what will allow us to speedup $\aorig$ to $O(\sqrt{n})$. 
\begin{lemma}
	\label{lem:radiuslowerbound}
	There exists some constant $\tc$ such that, if nodes $u$, $v$ of $T_o$ are at distance at least $\tc \sqrt{n}$ in $T$, then their corresponding nodes $\eta^{-1}(u)$ and $\eta^{-1}(v)$ are at distance at least $c \pumpmult \sqrt{n} /3$ in $S$.
\end{lemma}
\begin{proof}
	Consider a node $u$ at distance at least $\tnew$ from $v$ in $T$. There must exist a path $P$ in $T'$ connecting $\phi^{-1}(u)$ and $\phi^{-1}(v)$. By Lemma \ref{lem:deg2nodes}, at most $\frac{\sqrt{n}}{\tc}$ nodes in $P$ have degree greater than $2$, call the set of these nodes $X$.  We can bound the number of nodes of $P$ that are not part of paths that will be pumped in the following way:
	\begin{itemize}
		\item At most $\frac{c \sqrt{n} +1}{c+1} + \frac{\sqrt{n}}{c}+1$ nodes can be part of the ruling set. To see this, order the nodes of $P$ from left to right in one of the two canonical ways. The first summand bounds all the ruling set nodes whose right-hand short path is of length at least $\tc$, the second one bounds the ruling set nodes whose right-hand short path ends in a node $x \in X$, and the last one considers the path that ends in $\phi^{-1}(u)$ or $\phi^{-1}(v)$.
		\item At most $\frac{\sqrt{n}}{\tc} (1+2(\tc-1))$ nodes are either in $X$ or in short paths of length at most $\tc-1$ on the sides of a node in $X$.
		\item At most $2(c-1)$ nodes are between $\phi^{-1}(u)$ (or $\phi^{-1}(v)$) and a ruling set node.
	\end{itemize}
	While pumping the graph, in the worst case we replace paths of length $2 \tc $ with paths of length $\tc \pumpmult $, thus $\dist{\phi^{-1}(u)}{\phi^{-1}(v)}$ is at least
	\[
	\biggl(\tc \sqrt{n} +1 - \Bigl(\frac{c \sqrt{n} +1}{c+1} + \frac{\sqrt{n}}{c} + 1 + \frac{\sqrt{n}}{\tc} \bigl(1+2(\tc-1)\bigr) + 2(c-1)	\Bigr) \biggr) \cdot \frac{\tc \pumpmult }{2\tc} -1,
	\]
	which is greater than $c \pumpmult \sqrt n /3$ for $\tc$ and $n$ greater than a large enough constant.
\end{proof}

\subsection{Solving the problem faster}\label{sub:faster}
We now show how to speed up the algorithm $\aorig$ and obtain an algorithm running in $O(\sqrt{n})$. First, note that if the diameter of the original graph is $O(\sqrt{n})$, every node sees the whole graph in $O(\sqrt{n})$ rounds, and the problem is trivially solvable by brute force. Thus, in the following we assume that the diameter of the graph is $\omega(\sqrt{n})$. This also guarantees that $T_o$ is not empty.

Informally, nodes can distributedly construct the virtual tree $S$ in $O(\sqrt{n})$ rounds, and safely execute the original algorithm on it. Intuitively, even if a node $v$ sees just a part of $S$, we need to guarantee that this part has large enough radius, such that the original algorithm cannot see outside the subgraph of $S$ constructed by $v$ (otherwise $v$ would not be able to simulate the execution of $\aorig$ on $S$).

More precisely, all nodes do the following. First, they distributedly construct $S$, in $O(\sqrt{n})$ rounds. This increases the number of nodes, and requires nodes to assign new unique IDs to nodes that do not exist in the original graph (that is, nodes in the pumped regions). This new ID assignment can be computed in a standard manner as a function of the two IDs of the endpoints of the pumped paths. Then, each node $v$ in $T_o$ (nodes for which $\eta^{-1}(v)$ is defined), simulates the execution of $\aorig$ on node $\eta^{-1}(v)$ of $S$, by telling $\aorig$ that there are $N=\tc (\pumpmult+1) n$ nodes. Then, each node $v$ in $T_o$ outputs the same output assigned by $\aorig$ to node $\eta^{-1}(v)$ in $S$. Also, each node $v$ in $T_o$ fixes the output for all nodes in $\deltree_v$ ($\eta$ can be defined also for them, $v$ sees all of them, and the view of these nodes is contained in the view of $v$, thus it can simulate $\aorig$ in $S$ for all of them). Let $\Lambda$ be the set of nodes that already fixed an output, that is, $\Lambda =  \{ \{u\} \cup V(\deltree_u) ~|~ u \in T_o  \}$. Intuitively $\Lambda$ contains all the real nodes of $S$ (nodes with a corresponding node in $T$), including nodes removed when computing the skeleton tree, and leaves out only nodes that correspond to pumped regions. Finally, nodes in $V(T) \setminus \Lambda$ find a valid output via brute force. 

We need to prove two properties, the first shows that a node can safely execute $\aorig$ on the subgraph of $S$ that it knows, while the second shows that it is always possible to find a valid output for nodes in $V(T) \setminus \Lambda$ after having fixed outputs for nodes in $\Lambda$.

Let us choose a $\pumpmult$ satisfying $\told(N) \le c \pumpmult \sqrt{n} /3 $, where $\told(N)$ is the running time of $\aorig$. Note that $\pumpmult$ can be an arbitrarily large function of $n$. Such a $\pumpmult$ exists for all $\told(x) = o(x)$. We prove the following lemma.
\begin{lemma}
	For nodes in $T_o$, it is possible to execute $\aorig$ on $S$ by just knowing the neighbourhood of radius $2 \tc \sqrt{n}$ in $T$.
\end{lemma}
\begin{proof}
	First, note that by Lemma \ref{lem:newsizeupperbound}, the number of nodes of the virtual graph, $|V(S)|$, is always at most $N$, thus, it is not possible that a node of $S$ sees a number of nodes that is more than the number claimed when simulating the algorithm.
	
	Second, since $\pumpmult$ satisfies $\told(N) \le c \pumpmult \sqrt{n} /3$, and since, by Lemma  \ref{lem:radiuslowerbound} and the bound of $\tc \sqrt{n}$ on the depth of each deleted tree $\deltree_u$, the nodes outside a $2\tc \sqrt{n}$ ball of nodes in $T_o$ are at distance at least $\tc \pumpmult \sqrt{n} /3$ in $S$, the running time of $\aorig$ is less than the radius of the subtree of $S$ rooted at a node $v$ that $v$ distributedly computed and is aware of. This second part also implies that nodes in $T_o$ do not see the whole graph, thus they cannot notice that the value of $N$ is not the real size of the graph. 
\end{proof}

\subsection{Filling gaps by brute force}\label{sub:filling}
In this last part, we show that, by starting from a tree $T$ in which nodes of $\Lambda$ have already fixed an output, we can find a valid output for all the other nodes of the graph, in constant time. For this purpose, we adapt some definitions presented in \cite{Chang2019}, where it is shown that, by starting from a partially labelled graph, if we replace a subgraph with a different subgraph of the same \emph{type}, then the labelling of the original graph can be completed if and only if the labelling of the new graph can be completed. In our case the subgraphs that we replace are not labelled, and the following definitions handle exactly this case. In the following, unless stated otherwise, we use the term labelling to refer to an \emph{output} labelling.

We start by defining an equivalence relation $\eqrg$ between two pairs $(H,F)$ and $(H',F')$ composed of a graph and its poles. Intuitively, this equivalence relation says that equivalent $H$ and $H'$ should be isomorphic near the poles, and that if we fix some output near the poles of one graph, if we copy that output on the other graph (on the isomorphic part), and if that output is completable on the remaining nodes of the first graph, then it should be completable also on the other graph. A partial labelling (a partial function from nodes to labels) is called \emph{extendible} if it is possible to assign a label to unlabelled nodes such that it is locally consistent for every node, that is, the labelling satisfies the constraints of the given \lcl{} problem at every node.
This is a simplified version of the equivalence relation $\eqrg$ presented in \cite[Section 3.5]{Chang2019}.
\begin{definition}[The equivalence relation $\eqrg$]\label{def:eqrg}
	Given a graph $H$ and its poles $F$, define $\xi(H,F) = (D_1,D_2,D_3)$ to be a tripartition of $V(H)$ where 
	\begin{align*}
	D_1 &= \bigcup_{v\in F} N^{\checkradius-1}(v), \\
	D_2 &=  \bigcup_{v\in D_1} N^{\checkradius}(v) \setminus D_1, \\
	D_3 &= V(H) - (D_1 \cup D_2).
	\end{align*}
	Let $Q$ and $Q'$ be the subgraphs of $H$ and $H'$ induced by the vertices in $D_1 \cup D_2$ and $D'_1 \cup D'_2$ respectively.
	
	The equivalence holds, i.e., $(H,F) \eqrg (H',F')$, if and only if there is a $1$ to $1$ correspondence $\phi \colon (D_1 \cup D_2) \to (D'_1 \cup D'_2) $ satisfying:
	\begin{itemize}
		\item $Q$ and $Q'$ are isomorphic under $\phi$, preserving the input labels of the \lcl{} problem (if any), and preserving the order of the poles.
		\item Let $\labeling_*$ be any assignment of output labels to vertices in $D_1 \cup D_2$, and let $\labeling'_*$ be the corresponding labelling of $D'_1 \cup D'_2$ under $\phi$. Then $\labeling_*$ is extendible to $V(H)$ if and only if $\labeling'_*$ is extendible to $V(H')$. 
	\end{itemize} 
\end{definition}
In \cite{Chang2019}, it is proved that this equivalence relation is preserved after replacing equivalent subgraphs, and that, if the number of poles is constant, there is a constant number of equivalence classes.

Also, in \cite[Section 3.6]{Chang2019} the following lemma is proved. Informally, it shows that if we have a valid labelling for a graph, and we replace a subgraph with another equivalent to it, it is enough to change the labelling of the new subgraph  in order to obtain a valid labelling for the whole new graph. Also, the labelling near the borders is preserved. A labelling is locally consistent for node $v$ if the \lcl{} verifier running on node $v$ accepts that labelling.
\begin{lemma}
	\label{lem:thelemma}
	Let $G' = \replace(G,(H,F),(H',F'))$. Suppose $(H,F) \eqrg (H',F')$. Let $D_0 = V(G) \setminus V(H)$. Let $\labeling_{\diamond}$ be a complete labelling of $G$ that is locally consistent for all vertices in $D_2 \cup D_3$. Then there exists a complete labelling $\labeling'_{\diamond}$ satisfying the following:
	\begin{itemize}
		\item $\labeling_{\diamond}  = \labeling'_{\diamond}$ for all $v \in D_0 \cup D_1 \cup D_2$ and their corresponding vertices in $D'_0 \cup D'_1 \cup D'_2$. Also, if $\labeling_{\diamond}$ is locally consistent for a node $v$, then $\labeling'_{\diamond}$ is locally consistent for $\phi(v)$.
		\item $\labeling'_{\diamond}$ is locally consistent for all nodes in $D'_2 \cup D'_3$.
	\end{itemize}
	
\end{lemma}

We now adapt the definition of the function $\pump$ presented in \cite{Chang2019} for our purposes. Intuitively, as previously explained, starting from our tree $T$ we replace all subgraphs $Q^T \in \spaths^T$ with a pumped version of $Q^T$. Each $Q^T$ is composed of a main path in which, for each node, there is a subtree of height $O(\sqrt{n})$. Note that $Q^T$ is connected to the rest of $T$ on the two endpoints of the main path, thus it has two poles, that implies, as previously discussed, that the number of equivalence classes under $\eqrg$ is constant. This class is also computable by a node, since it considers only the subtrees $Q^T$ that are contained in its ball. Also, we can see $Q^T$ as a sequence of $\deltree_v$, and the type of a $Q^T$ can be computed, as in \cite{Chang2019}, by reading one ``character'' (class of $\deltree_v$) at a time.
Finally, we can see the sequence as a string that, if it is long enough, we can pump in order to obtain a longer string of the same type. More formally, consider a tree $Q^T \in \spaths^T$. We can see $Q^T$ as a path of length $k$, where each node $i$ is the root of a tree $\deltree_i$ ($1\le i \le k$). Let $(\deltree_i)_{i \in [k]}$ denote this path. Let $\typetree(\deltree_j)$ be the equivalence class of the tree $\deltree_j$ considering $j$ as the unique pole, and let $\typepath(H)$ be the equivalence class of the path $H$ considering its endpoints as poles.

The following lemma says that nodes can compute the type of the deleted trees rooted on nodes contained in their balls.
\begin{lemma}\label{lem:knowntypes}
	Each node $u$ can determine the type of $\deltree_v$ for all $v \in \ball{u}$. 
\end{lemma}
\begin{proof}
	When nodes compute the skeleton tree $T'$, they broadcast all their balls to the nodes inside their balls. Since a tree $\deltree_v$ has height $O(\sqrt{n})$, it is fully contained in the ball of $v$, thus all the nodes in the ball of $v$ see the whole tree $\deltree_v$, and can determine its type (it depends only on the structure of $\deltree_v$ and the inputs of the nodes in this tree).
\end{proof}

The following is a crucial lemma proved in \cite[Section 3.8]{Chang2019}.
\begin{lemma}\label{lem:automata}
	Let $H = (\deltree_i)_{i \in [k]}$ and $H' = (\deltree_i)_{i \in [k+1]}$ be identical to $H$ in its first $k$ trees. Then $\typepath(H')$ is a function of $\typepath(H)$ and $\typetree(\deltree_{k+1})$. 
\end{lemma}

As shown in \cite{Chang2019}, Lemma \ref{lem:automata} allows us to bring classic automata theory into play. By Lemma \ref{lem:knowntypes}, nodes can know the type of each $\deltree_i$ contained in a path that they want to pump. Consider a path $H = (\deltree_i)_{i \in [k]}$, and the sequence $C = (c_1,\ldots,c_k)$, where $c_i$ is $\typetree(\deltree_i)$. A finite automaton can determine the type of $H$ by reading one character of $C$ at a time. The number of states in this automaton is constant, let $\lpump$ be such a constant. The following lemma holds \cite[Lemma 7]{Chang2019}).

\begin{lemma}\label{lem:pump}
	Let $H = (\deltree_i)_{i \in [k]}$, with $k \ge \lpump$. $H$ can be decomposed into three subpaths $H = x \circ y \circ z$ such that:
	\begin{itemize}
		\item $|xy| \le \lpump$,
		\item $|y| \ge 1$,
		\item $\typepath(x \circ y^j \circ z) = \typepath(H)$ for each non-negative $j$. 
	\end{itemize}
\end{lemma}

We finally define the function $\pump$, that, given a tree $Q^T$ having a main path of short length, produces a new tree $P^T$ having a main path that is arbitrary longer, such that their types are equivalent.

\begin{definition}[$\pump$]
	Let $Q^T \in \spaths^T$, and fix $\tc = \lpump + 4 \checkradius$. We have that the main path of $Q^T$ has length at least $\lpump  + 4 \checkradius$. Let us split the main path of $Q^T$ in three subpaths $p_l, p_c, p_r$, two of length $2\checkradius$ near the poles ($p_l$ and $p_r$), and one of length at least $\lpump$ containing the remaining nodes ($p_c$). $\pump(Q^T,\pumpmult)$ produces a tree $P^T$ such that $\typepath(Q^T) = \typepath(P^T)$ and the main path of $P^T$ has length $l'$ satisfying $c B \le l' \le c (B+1)$. This is obtained by pumping the subpath $p_c$.  By Lemma \ref{lem:pump} such a function exists. Since the paths $p_l$ and $p_r$ are preserved during the pump, the isomorphism near the poles is preserved.
\end{definition}

We now prove that the partial labelling produced by the algorithm previously described can be completed consistently. Consider the tripartition described in Definition \ref{def:eqrg}. Let $\mathcal{R}$ be the union of all the replaced subgraphs, and let $D_1$, $D_2$, $D_3$ be a tripartition of it as defined in Definition \ref{def:eqrg}. By definition of $T_o$, $\Lambda$ corresponds to nodes in $D_0 \cup D_1 \cup D_2$.

First, notice that a node in $\mathcal{R}$ sees all the nodes in the regions $D_1$ and $D_2$ of the replaced subgraph where it is located, thus it has enough information needed to find a valid output via brute force.

Second, by Lemma \ref{lem:thelemma}, in order to show that the partial labelling can be completed consistently, it is enough to show that each replaced $Q^T$ is in the same equivalence class as $\pump(Q^T,\pumpmult)$, which is true by the definition of $\pump$.

\bibliographystyle{plainurl}
\bibliography{bibliography}

\begin{thebibliography}{10}

\bibitem{Balliu2018stoc}
Alkida Balliu, Juho Hirvonen, Janne~H Korhonen, Tuomo Lempi{\"{a}}inen, Dennis
  Olivetti, and Jukka Suomela.
\newblock {New classes of distributed time complexity}.
\newblock In {\em Proc. 50th ACM Symposium on Theory of Computing (STOC 2018)},
  pages 1307--1318. ACM Press, 2018.
\newblock \href {http://dx.doi.org/10.1145/3188745.3188860}
  {\path{doi:10.1145/3188745.3188860}}.

\bibitem{barenboim16sublinear}
Leonid Barenboim.
\newblock {Deterministic ($\Delta$+1)-Coloring in Sublinear (in $\Delta$) Time
  in Static, Dynamic, and Faulty Networks}.
\newblock {\em Journal of the ACM}, 63(5):1--22, 2016.
\newblock \href {http://dx.doi.org/10.1145/2979675}
  {\path{doi:10.1145/2979675}}.

\bibitem{barenboim14distributed}
Leonid Barenboim, Michael Elkin, and Fabian Kuhn.
\newblock {Distributed ($\Delta$+1)-Coloring in Linear (in $\Delta$) Time}.
\newblock {\em SIAM Journal on Computing}, 43(1):72--95, 2014.
\newblock \href {http://dx.doi.org/10.1137/12088848X}
  {\path{doi:10.1137/12088848X}}.

\bibitem{Brandt2016}
Sebastian Brandt, Orr Fischer, Juho Hirvonen, Barbara Keller, Tuomo
  Lempi{\"{a}}inen, Joel Rybicki, Jukka Suomela, and Jara Uitto.
\newblock {A lower bound for the distributed Lov{\'{a}}sz local lemma}.
\newblock In {\em Proc. 48th ACM Symposium on Theory of Computing (STOC 2016)},
  pages 479--488. ACM Press, 2016.
\newblock \href {http://dx.doi.org/10.1145/2897518.2897570}
  {\path{doi:10.1145/2897518.2897570}}.

\bibitem{Brandt2017}
Sebastian Brandt, Juho Hirvonen, Janne~H Korhonen, Tuomo Lempi{\"{a}}inen,
  Patric R~J {\"{O}}sterg{\aa}rd, Christopher Purcell, Joel Rybicki, Jukka
  Suomela, and Przemys{\l}aw Uzna{\'{n}}ski.
\newblock {LCL problems on grids}.
\newblock In {\em Proc. 36th ACM Symposium on Principles of Distributed
  Computing (PODC 2017)}, pages 101--110. ACM Press, 2017.
\newblock \href {http://dx.doi.org/10.1145/3087801.3087833}
  {\path{doi:10.1145/3087801.3087833}}.

\bibitem{chang18complexity}
Yi-Jun Chang, Qizheng He, Wenzheng Li, Seth Pettie, and Jara Uitto.
\newblock Distributed edge coloring and a special case of the constructive
  {L}ov{\'{a}}sz local lemma.
\newblock {\em {ACM} Transactions on Algorithms}, 16(1):8:1--8:51, 2020.
\newblock \href {http://dx.doi.org/10.1145/3365004}
  {\path{doi:10.1145/3365004}}.

\bibitem{chang16exponential}
Yi-Jun Chang, Tsvi Kopelowitz, and Seth Pettie.
\newblock An exponential separation between randomized and deterministic
  complexity in the {LOCAL} model.
\newblock {\em {SIAM} Journal on Computing}, 48(1):122--143, 2019.
\newblock \href {http://dx.doi.org/10.1137/17M1117537}
  {\path{doi:10.1137/17M1117537}}.

\bibitem{Chang2019}
Yi-Jun Chang and Seth Pettie.
\newblock {A Time Hierarchy Theorem for the LOCAL Model}.
\newblock {\em {SIAM} Journal on Computing}, 48(1):33--69, 2019.
\newblock \href {http://dx.doi.org/10.1137/17M1157957}
  {\path{doi:10.1137/17M1157957}}.

\bibitem{cole86deterministic}
Richard Cole and Uzi Vishkin.
\newblock {Deterministic coin tossing with applications to optimal parallel
  list ranking}.
\newblock {\em Information and Control}, 70(1):32--53, 1986.
\newblock \href {http://dx.doi.org/10.1016/S0019-9958(86)80023-7}
  {\path{doi:10.1016/S0019-9958(86)80023-7}}.

\bibitem{fraigniaud16local}
Pierre Fraigniaud, Marc Heinrich, and Adrian Kosowski.
\newblock {Local Conflict Coloring}.
\newblock In {\em Proc. 57th IEEE Annual Symposium on Foundations of Computer
  Science (FOCS 2016)}, pages 625--634. IEEE, 2016.
\newblock \href {http://dx.doi.org/10.1109/FOCS.2016.73}
  {\path{doi:10.1109/FOCS.2016.73}}.

\bibitem{GHK18}
Mohsen Ghaffari, David~G. Harris, and Fabian Kuhn.
\newblock On derandomizing local distributed algorithms.
\newblock In {\em Proc. 59th {IEEE} Annual Symposium on Foundations of Computer
  Science (FOCS 2018)}, pages 662--673. IEEE, 2018.
\newblock \href {http://dx.doi.org/10.1109/FOCS.2018.00069}
  {\path{doi:10.1109/FOCS.2018.00069}}.

\bibitem{ghaffari17distributed}
Mohsen Ghaffari and Hsin-Hao Su.
\newblock {Distributed Degree Splitting, Edge Coloring, and Orientations}.
\newblock In {\em Proc. 28th ACM-SIAM Symposium on Discrete Algorithms (SODA
  2017)}, pages 2505--2523. Society for Industrial and Applied Mathematics,
  2017.
\newblock \href {http://dx.doi.org/10.1137/1.9781611974782.166}
  {\path{doi:10.1137/1.9781611974782.166}}.

\bibitem{Goos2016}
Mika G{\"{o}}{\"{o}}s and Jukka Suomela.
\newblock {Locally checkable proofs in distributed computing}.
\newblock {\em Theory of Computing}, 12, 2016.
\newblock \href {http://dx.doi.org/10.4086/toc.2016.v012a019}
  {\path{doi:10.4086/toc.2016.v012a019}}.

\bibitem{HS95}
Juris Hartmanis and Richard~Edwin Stearns.
\newblock {On the computational complexity of algorithms}.
\newblock {\em Transactions of the American Mathematical Society},
  117(117):285--285, 1965.
\newblock \href {http://dx.doi.org/10.1090/S0002-9947-1965-0170805-7}
  {\path{doi:10.1090/S0002-9947-1965-0170805-7}}.

\bibitem{HU79}
John~E Hopcroft and Jeffrey~D Ullman.
\newblock {\em {Introduction to Automata Theory, Languages and Computation}}.
\newblock Addison-Wesley, 1979.

\bibitem{Linial1992}
Nathan Linial.
\newblock {Locality in Distributed Graph Algorithms}.
\newblock {\em SIAM Journal on Computing}, 21(1):193--201, 1992.
\newblock \href {http://dx.doi.org/10.1137/0221015}
  {\path{doi:10.1137/0221015}}.

\bibitem{Naor1995}
Moni Naor and Larry Stockmeyer.
\newblock {What Can be Computed Locally?}
\newblock {\em SIAM Journal on Computing}, 24(6):1259--1277, 1995.
\newblock \href {http://dx.doi.org/10.1137/S0097539793254571}
  {\path{doi:10.1137/S0097539793254571}}.

\bibitem{panconesi01simple}
Alessandro Panconesi and Romeo Rizzi.
\newblock {Some simple distributed algorithms for sparse networks}.
\newblock {\em Distributed Computing}, 14(2):97--100, 2001.
\newblock \href {http://dx.doi.org/10.1007/PL00008932}
  {\path{doi:10.1007/PL00008932}}.

\bibitem{panconesi95delta}
Alessandro Panconesi and Aravind Srinivasan.
\newblock {The local nature of $\Delta$-coloring and its algorithmic
  applications}.
\newblock {\em Combinatorica}, 15(2):255--280, 1995.
\newblock \href {http://dx.doi.org/10.1007/BF01200759}
  {\path{doi:10.1007/BF01200759}}.

\bibitem{Peleg2000}
David Peleg.
\newblock {\em {Distributed Computing: A Locality-Sensitive Approach}}.
\newblock Society for Industrial and Applied Mathematics, 2000.
\newblock \href {http://dx.doi.org/10.1137/1.9780898719772}
  {\path{doi:10.1137/1.9780898719772}}.

\end{thebibliography}

\end{document}